\documentclass[11pt,a4paper]{amsart}
\usepackage{bbm}
\usepackage{stmaryrd}
\usepackage[pdftex]{hyperref,color,graphicx}
\newtheorem{theorem}{Theorem}[section]

\newtheorem{lemma}[theorem]{Lemma}

\newtheorem{proposition}[theorem]{Proposition}

\theoremstyle{definition}
\newtheorem{definition}[theorem]{Definition}

\newtheorem{fact}[theorem]{Fact}

\theoremstyle{remark}

\newtheorem{remark}[theorem]{Remark}
\newtheorem{example}[theorem]{Example}

\numberwithin{equation}{section}

\newcommand{\Leb}{\mathrm{Leb}}
\newcommand{\conv}{\mathrm{conv}}

\newcommand{\ud}{\,\mathrm{d}}
\newcommand{\e}{\mathrm{e}}
\newcommand{\Exp}{\mathrm{Exp}}

\DeclareMathOperator*{\esssup}{ess\,sup}

\begin{document}
%
%
%
%
\title[Utility from inter-temporal wealth]{Infinite horizon utility
  maximisation from inter-temporal wealth}

\author[Michael Monoyios]{Michael Monoyios} 

\address{Mathematical Institute \\ 
University of Oxford \\
Radcliffe Observatory Quarter\\
Woodstock Road\\
Oxford OX2 6GG\\
UK}

\email{monoyios@maths.ox.ac.uk}

\date{\today}

\thanks{Part of this work was carried out during a visit to the
  Laboratoire de Probabilit\'es et Mod\`eles Al\'eatoires,
  Universit\'e Paris Diderot. I am very grateful to Huy\^en Pham
  for generous hospitality.}

\begin{abstract}

We develop a duality theory for the problem of maximising expected
lifetime utility from inter-temporal wealth over an infinite horizon,
under the minimal no-arbitrage assumption of No Unbounded Profit with
Bounded Risk (NUPBR). We use only deflators, with no arguments
involving equivalent martingale measures, so do not require the
stronger condition of No Free Lunch with Vanishing Risk (NFLVR). Our
formalism also works without alteration for the finite horizon version
of the problem. As well as extending work of Bouchard and Pham
\cite{bp04} to any horizon and to a weaker no-arbitrage setting, we
obtain a stronger duality statement, because we do not assume by
definition that the dual domain is the polar set of the primal
space. Instead, we adopt a method akin to that used for inter-temporal
consumption problems, developing a supermartingale property of the
deflated wealth and its path that yields an infinite horizon budget
constraint and serves to define the correct dual variables. The
structure of our dual space allows us to show that it is convex,
without forcing this property by assumption. We proceed to enlarge the
primal and dual domains to confer solidity to them, and use
supermartingale convergence results which exploit Fatou convergence,
to establish that the enlarged dual domain is the bipolar of the
original dual space. The resulting duality theorem shows that all the
classical tenets of convex duality hold. Moreover, at the optimum, the
deflated wealth process is a potential converging to zero. We work out
examples, including a case with a stock whose market price of risk is
a three-dimensional Bessel process, so satisfying NUPBR but not NFLVR.

\end{abstract}

\maketitle

\tableofcontents

\section{Introduction}
\label{sec:intro}

Let $U:\mathbb{R}_{+}\to\mathbb{R}$ be a classical utility function
and $(X_{t})_{t\geq 0}$ a non-negative wealth process generated from
self-financing investment in a semimartingale incomplete market on a
complete stochastic basis
$(\Omega,\mathcal{F},\mathbb{F}:=(\mathcal{F}_{t})_{t\in[0,\infty)},
\mathbb{P})$, with the filtration $\mathbb{F}$ satisfying the usual
hypotheses of right-continuity and augmentation with $\mathbb{P}$-null
sets of $\mathcal{F}$. Under the minimal no-arbitrage assumption of No
Unbounded Profit with Bounded Risk (NUPBR), we develop a duality
theory for a problem in which utility is derived from inter-temporal
wealth over the infinite horizon:
\begin{equation}
\mathbb{E}\left[\int_{0}^{\infty}U(X_{t})\ud\kappa_{t}\right] \to \max !
\label{eq:basicproblem}
\end{equation}
In \eqref{eq:basicproblem}, $\kappa:[0,\infty)\to\mathbb{R}_{+}$ is a
non-decreasing c\`adl\`ag adapted process that will act as a finite
measure to assign a weight to utility of wealth at each time. We focus
on the infinite horizon case, but our approach also works without
alteration for the finite horizon version of \eqref{eq:basicproblem},
as we re-iterate in Remark \ref{rem:fh}.

Problems of the type in \eqref{eq:basicproblem} can arise when
traditional utility of terminal wealth problems have a random horizon
date, as we shall illustrate by some examples in Section
\ref{subsubsec:sexamples}, but can just as well be considered in their
own right as one possible objective for a long-lived investment
fund. A duality theory for such problems was developed by Bouchard and
Pham \cite{bp04} over a finite horizon, with a no-arbitrage assumption
that allowed for the existence of equivalent local martingale measures
(ELMMs), so tantamount to assuming No Free Lunch with Vanishing Risk
(NFLVR) in the terminology of Delbaen and Schachermayer
\cite{ds94}. Here, the underlying assumptions as well as the approach
and construction of the dual space are different to those in
\cite{bp04}, as we now describe.

First, as indicated above, we relax the no-arbitrage assumption from
NFLVR to NUPBR, so we do not rely on the existence of ELMMs, only on
the existence of a class of deflators that multiply admissible wealth
processes to create supermartingales. It was first made explicit by
Karatzas and Kardaras \cite{kk07} (though was implicit in the terminal
wealth problem of Karatzas et al \cite{klsx91} in an incomplete It\^o
process market, in which which ELMMs were not invoked at all) that all
one needs for well-posed utility maximisation problems is the
existence of a suitable class of deflators to act as dual
variables. In particular, ELMMs are not needed. This is a first reason
for adopting NUPBR as our no-arbitrage condition.

Aside from weakening the no-arbitrage assumption, there are other
sound reasons for avoiding the use of ELMMs. It is well known that
ELMMs will typically not exist over the infinite horizon, because the
candidate change of measure density process is not a uniformly
integrable martingale. This is the case for the Black-Scholes model
for example, as discussed in Karatzas and Shreve \cite[Section
1.7]{ks98}. Moreover, even if ELMMs might exist when restricted to a
finite horizon, one needs to proceed with some care in invoking them
in an infinite horizon model, by ensuring that events in the tail
$\sigma$-algebra
$\mathcal{F}_{\infty}:=\sigma\left(\bigcup_{t\geq0}\mathcal{F}_{t}\right)$
have been excluded in a consistent way. We discuss this issue further
in Section \ref{subsubsec:completion}. Irrespective of such
subtleties, since deflators are the key ingredient for establishing a
duality for utility maximisation problems, it is natural to construct
a theory which uses only deflators, and makes no use whatsoever of
constructions involving ELMMs, and this is what we do. A key step in
this approach will be the use of the Stricker and Yan \cite{sy98}
version of the Optional Decomposition Theorem (ODT) to establish
bipolarity relations between the primal and dual domains, as opposed
to variants of the ODT which state the result in terms of ELMMs.

Second, our approach to establishing the duality between the primal
problem in \eqref{eq:basicproblem} and an appropriately defined dual
problem differs quite markedly from that in Bouchard and Pham
\cite{bp04}, and our basic duality statement is strengthened compared
to that in \cite{bp04}, in essence because we are able to prove, as
opposed to assume by definition, that the dual domain is the polar of
the primal domain, as we now describe.

The approach taken in \cite{bp04}, over a finite horizon time $[0,T]$,
is to define the dual domain (in the case where the initial value of
the dual variables is unity) as the set of processes $Y$ such that
$\mathbb{E}\left[\int_{0}^{T}X_{t}Y_{t}\ud\kappa_{t}\right]\leq 1$ for
all admissible wealth processes with unit initial capital. In other
words, the dual domain was explicitly defined in \cite{bp04} as the
polar of the primal domain. This automatically renders the dual domain
convex and closed, so bypasses some steps in establishing bipolarity
relations between the primal and dual spaces, and hence the duality
theorem, but at the expense of weakening the final statement to some
degree. This is also the reason that the bulk of the remaining
analysis in \cite{bp04} takes place in the primal domain.

In our method, by contrast, we find the form of the dual problem and
the associated dual domain by seeking a supermartingale property
satisfied by the pair $(X_{t},(X_{s})_{0\leq s\leq t})_{t\geq 0}$,
that is, the value of an admissible wealth process at any time, as
well as the wealth path up to that time, as follows. Let $S$ be any
classical supermartingale deflator, so $XS$ is a supermartingale for
all admissible wealth processes, and let $\beta$ be a non-negative
process such that $\int_{0}^{\cdot}\beta_{s}\ud\kappa_{s}$ is almost
surely finite. We define associated supermartingales $R$ and processes
$Y$ by
\begin{equation} 
R := \exp\left(-\int_{0}^{\cdot}\beta_{s}\ud\kappa_{s}\right)S, \quad  
Y := \beta R =
\beta\exp\left(-\int_{0}^{\cdot}\beta_{s}\ud\kappa_{s}\right)S. 
\label{eq:RYintro}
\end{equation}
With these processes in place, we show that
$M:=XR+\int_{0}^{\cdot}X_{s}Y_{s}\ud\kappa_{s}$ is a supermartingale
for all admissible wealth processes. The wealth-path deflators $Y$ are
then the appropriate dual variables for the problem in
\eqref{eq:basicproblem}. They involve the auxiliary dual control
$\beta$ above and beyond that implicit in the choice of
supermartingale deflator, a typical feature of wealth path dependent
utility maximisation problems.

This program yields an infinite horizon budget constraint satisfied by
the wealth path, similar to that in Bouchard and Pham \cite{bp04},
over our infinite horizon:
$\mathbb{E}\left[\int_{0}^{\infty}X_{t}Y_{t}\ud\kappa_{t}\right]\leq
1$ for all admissible wealth processes with unit initial capital and
all deflators with unit initial value. The budget constraint so formed
acts (at this point) as a necessary condition for admissible wealth
processes and serves to define the appropriate dual variables $Y$. The
form of the dual problem then emerges as
\begin{equation*}
\mathbb{E}\left[\int_{0}^{\infty}V(Y_{t})\ud\kappa_{t}\right] \to \min!  
\end{equation*}
over deflators with initial value $Y_{0}=y>0$, where
$V:\mathbb{R}_{+}\to\mathbb{R}$ is the convex conjugate of the utility
function.

The particular structure in \eqref{eq:RYintro} of the inter-temporal
wealth deflators, involving the supermartingale deflators and the
auxiliary dual control $\beta$, is crucial, as it allows us to show
that the dual space which emerges is convex. We then enlarge the
primal domain to encompass processes dominated by admissible wealths
(similar in spirit to the procedure used by Kramkov and Schachermayer
\cite{ks99,ks03} for the terminal wealth utility maximisation
problem), and show that the budget constraint is also a sufficient
condition for admissible primal processes, using the Stricker and Yan
\cite{sy98} version of the Optional Decomposition Theorem. Finally, we
enlarge the dual domain in a similar manner, to encompass processes
dominated by the deflators, show that the resulting dual domain is
closed in an appropriate topology (that of convergence in measure
$\mu:=\kappa\times\mathbb{P}$) by exploiting Fatou convergence of
supermartingales, and obtain perfect bipolarity relations between the
enlarged primal and dual domains. This bipolarity underlies the
subsequent duality results.

We thus prove (as opposed to impose, by definition) that our dual
domain has the required convexity and closedness properties needed to
establish bipolarity and hence duality, with a supermartingale
constraint involving the admissible wealths as a starting point. Put
another way, the procedure developed by Kramkov and Schachermayer
\cite{ks99,ks03} for the terminal wealth problem is adapted and made
to work for an inter-temporal wealth problem under NUPBR and over the
infinite (or indeed, finite) horizon.

The main duality result (Theorem \ref{thm:itwd}) shows that all the
tenets of the theory hold in our scenario: the marginal utility of
optimal wealth is equal to the optimal deflator with initial value
equal to the derivative of the primal value function, and the primal
and dual value functions are mutually conjugate. Moreover, at the
optimum, the supermartingale $M$ becomes a uniformly integrable
martingale $\widehat{M}$, leading to an interesting additional
representation of the optimal wealth process:
\begin{equation}
\widehat{X}_{t}\widehat{R}_{t} = \mathbb{E}\left[\left.\int_{t}^{\infty}
\widehat{X}_{s}\widehat{Y}_{s}\ud\kappa_{s}\right\vert
\mathcal{F}_{t}\right], \quad t\geq 0,
\label{eq:owpintro}
\end{equation}
where $\widehat{X},\widehat{R},\widehat{Y}$ are the optimal
manifestations of the processes $X,R,Y$. The supermartingale $XR$
becomes a potential (satisfying
$\lim_{t\to\infty}\mathbb{E}[\widehat{X}_{t}\widehat{R}_{t}]=0$) at
the optimum, and also converges almost surely to
$\widehat{X}_{\infty}\widehat{R}_{\infty}=0$.

The philosophy of our approach can thus be summarised as: the use of a
natural supermartingale property to derive an inter-termporal budget
constraint as a means of identifying the dual space, then suitably
enlarging that space and using supermartingale convergence techniques
to get the bipolarity relations. A similar philosophy was applied
recently by Monoyios \cite{mmc20} to the infinite horizon optimal
consumption problem under NUPBR.  There, using an appropriate class of
consumption deflators (differing from the wealth-path deflators used
here), deflated wealth plus cumulative deflated consumption at the
optimum becomes a uniformly integrable martingale, while deflated
wealth becomes a potential converging to zero. The final results in
\cite{mmc20} thus have a similar flavour to those here (as they must,
since classical tenets of duality theory are shoiwn to hold in both
cases), but the problem studied in \cite{mmc20} is quite distinct from
the one here, involving a different primal variable (consumption, as
opposed to the wealth path), and a different set of concomitant
deflators.  It turns out that the primal domain in \cite{mmc20} is
$L^{1}(\mu)$-bounded, while here it is the dual space that has this
property. We thus do not use any results from \cite{mmc20} here, as we
need to formulate a separate, complete proof.

Inter-temporal utility maximisation problems have a long history,
usually treating the problem of maximising utility of consumption, as
opposed to the less standard objective here, which involves aggregate
expected utility from inter-temporal wealth, in the absence of a
consumption stream. The consumption literature begins with Merton's
\cite{merton69} dynamic programming solution of the problem in a
constant coefficient complete Brownian model, extended to cover issues
such as non-negativity constraints on consumption, and bankruptcy, by
Karatzas et al \cite{kal86} using similar methods. The infinite
horizon problem for utility from consumption in a complete It\^o
market was treated via duality methods by Huang and Pag\`es
\cite{hp92}, while Foldes \cite{foldes90,foldes92} characterised
optimal consumption plans in semimartingale markets with well-defined
``shadow prices'' (local martingales that characterise marginal
utility of consumption processes). These correspond to some class of
deflators in the modern mathematical finance terminology. The set-up
of \cite{foldes90,foldes92} is very much rooted in traditional
economic phraseology, so certain contemporary notions such as the
underlying no-arbitrage condition, and the completeness or otherwise
of the market, are not entirely transparent. Karatzas and
\v{Z}itkovi\'{c} \cite{kz03} and \v{Z}itkovi\'{c} \cite{z05} treated
problems of optimal consumption with an additional random endowment,
in incomplete semimartingale markets on a finite horizon and under the
classical no-arbitrage assumption of the existence of ELMMs,
equivalent to NFLVR. The consumption problem in an infinite horizon
semimartingale incomplete market under NFLVR was treated by Mostovyi
\cite{most15}, and later by Chau et al \cite{chauetal17} under NUPBR,
with the recent treatment in \cite{mmc20} establishing the duality
without recourse to any arguments involving ELMMs.

In contrast to the consumption problems analysed by these papers, the
problem studied here, of maximising utility from inter-temporal
wealth, has received much less attention. As indicated earlier, a dual
theory for such problems, over a finite horizon and under NFLVR, was
developed by Bouchard and Pham \cite{bp04}. Aside from the duality
developed in \cite{bp04}, wealth-path-dependent utility maximisation
problems have arisen in models which consider investment and
consumption with a random horizon, such as Blanchet-Scalliet et al
\cite{bsekmm08} (in complete Brownian markets with deterministic
parameters), or Vellekoop and Davis \cite{dv09} (who consider a
Merton-type problem of optimal consumption in a Black-Scholes model,
but with randomly terminating income). Federico et al \cite{fgg15}
analyse wealth-path-dependent problems from the dynamic programming
and Hamilton-Jacobi-Bellman (HJB) equation viewpoint, using viscosity
solution methods to establish regularity of the value functions in
Markovian market scenarios driven by Brownian motions. As well as
arising from a random termination date, the problem studied here can
also be viewed in its own right as describing an objective for a
long-lived investment fund, looking to build, as oppposed to consume,
wealth.

The rest of the paper is structured as follows. In Section
\ref{sec:fmpf} we describe the financial market, introduce various
classes of deflators and the primal problem, list some examples which
fit into our set-up, then derive the budget constraint and formulate
the dual problem. In Section \ref{sec:tmd} we give the main duality
theorem (Theorem \ref{thm:itwd}), and describe how the result may be
re-cast in the case when $\kappa$ is absolutely continuous with
respect to Lebesgue measure (Remark \ref{rem:sc}). In Section
\ref{sec:abpd} we formulate the primal and dual problems in abstract
notation on a finite measure space with product measure
$\mu:=\kappa\times\mathbb{P}$. We re-cast the optimisation problems
over suitably enlarged primal and dual domains, and present the 
bipolarity relations between these spaces (Proposition \ref{prop:abp})
as well as the abstract version of the duality theorem (Theorem
\ref{thm:adt}). In Section \ref{sec:bpr} we prove Proposition
\ref{prop:abp}. In many respects this is the heart of the paper. We
use the Stricker and Yan \cite{sy98} optional decomposition results to
show that the budget constraint is also a sufficient condition for
primal admissibility, then show that the dual domain we have
constructed is convex and closed, and make comparisons with the
approach of Bouchard and Pham \cite{bp04}. In Section \ref{sec:pdt} we
prove the abstract duality theorem in the classical manner of Kramkov
and Schachermayer \cite{ks99,ks03}, from which the concrete duality
theorem follows, and also prove the novel representation
\eqref{eq:owpintro} of the optimal wealth process (Proposition
\ref{prop:owp}). In Section \ref{sec:examples} we work out two
examples with power and logarithmic utility: a model whose market
price of risk is a three-dimensional Bessel process (so satisfying
NUPBR but not NFLVR) with stochastic volatility and correlation, and a
Black-Scholes market.

\section{Financial market and problem formulation}
\label{sec:fmpf}

\subsection{The financial market}
\label{subsec:fm}

We have an infinite horizon financial market on a complete stochastic
basis
$(\Omega,\mathcal{F},\mathbb{F}:=(\mathcal{F}_{t})_{t\in[0,\infty)},
\mathbb{P})$, with the filtration $\mathbb{F}$ satisfying the usual
hypotheses of right-continuity and augmentation with $\mathbb{P}$-null
sets of $\mathcal{F}$. The market contains $d$ stocks and a cash
asset, the latter with strictly positive price process. We shall use
the cash asset as num\'eraire, so without loss of generality (as we
shall affirm in Remark \ref{rem:du}) its price is normalised to unity
and we work with discounted quantities throughout. The (discounted)
price processes of the stocks are given by a non-negative c\`adl\`ag
vector semimartingale $P=(P^{1},\ldots,P^{d})$.

The $\sigma$-algebra $\mathcal{F}$ can contain more information than
that generated by the asset prices, so can include, for example, a
random time at which investment ceases, as this is one scenario where
inter-temporal wealth utility maximisation can arise. Bouchard and
Pham \cite{bp04} had a similar feature in a finite horizon version of
our utility maximisation problem under NFLVR. Note that our formalism
and results can be transferred with no alteration to the finite
horizon setting, as we re-iterate in Remark \ref{rem:fh}.

A financial agent can trade a self-financing portfolio of the stocks
and cash. The agent has initial capital $x>0$, with the trading
strategy represented by a $d$-dimensional predictable $P$-integrable
process $H=(H^{1},\ldots,H^{d})$, with $H^{i},\,i=1,\ldots,d$ the
process for the number of shares of the $i^{\mathrm{th}}$ stock in the
portfolio. The agent's wealth process $X$ is given by
\begin{equation*}
X_{t} := x + (H\cdot P)_{t}, \quad t\geq 0, \quad x>0,
\end{equation*}
where $(H\cdot P):=\int_{0}^{\cdot}H_{s}\ud P_{s}$ denotes the
stochastic integral. Let $\mathcal{X}(x)$ denote the
set of non-negative wealth processes with initial wealth $x>0$:
\begin{equation*}
\mathcal{X}(x) := \{X: X=x+(H\cdot P)\geq 0,\, \mbox{almost surely}\},
\quad x>0.
\end{equation*}
We write $\mathcal{X}\equiv\mathcal{X}(1)$ and we have
$\mathcal{X}(x)=x\mathcal{X}=\{xX:X\in\mathcal{X}\}$ for $x>0$. The
set $\mathcal{X}$ is a convex (and hence so is
$\mathcal{X}(x),\,x>0$).

For $y>0$, let $\mathcal{S}(y)$ denote the set of
\textit{supermartingale deflators} (SMDs), positive c\`adl\`ag
processes $S$ with $S_{0}=y$ such that the deflated wealth $SX$ is a
supermartingale for all $X\in\mathcal{X}$:
\begin{equation}
\mathcal{S}(y) := \left\{S>0,\,\mbox{c\`adl\`ag},\, S_{0}=y:
\mbox{$SX$ is a supermartingale for all
$X\in\mathcal{X}$}\right\}.
\label{eq:smd}
\end{equation} 
We write $\mathcal{S}\equiv\mathcal{S}(1)$, and we have
$\mathcal{S}(y)=y\mathcal{S}$ for $y>0$. The set $\mathcal{S}$ is
clearly convex. Since the constant process $X\equiv 1$ lies in
$\mathcal{X}$, each $S\in\mathcal{S}$ is a supermartingale. The
supermartingale deflators are the processes used as dual variables by
Kramkov and Schachermayer \cite{ks99,ks03} in their treatment of the
terminal wealth utility maximisation problem. The dual domain for the
forthcoming inter-temporal wealth problem will be based on
$\mathcal{S}(y)$ but will not coincide with this space, as we shall
see shortly.

Let $\mathcal{Z}$ denote the set of \textit{local martingale
deflators} (LMDs), positive c\`adl\`ag local martingales $Z$ with
unit initial value such that deflated wealth $XZ$ is a local
martingale for all $X\in\mathcal{X}$:
\begin{equation}
\mathcal{Z} := \left\{Z>0,\,\mbox{c\`adl\`ag},\, Z_{0}=1:
\mbox{$XZ$ is a local martingale for all
$X\in\mathcal{X}$}\right\}. 
\label{eq:mcZ}
\end{equation}
Since the local martingale $XZ\geq 0$ for all $X\in\mathcal{X}$, it is
also a supermartingale and, since $X\equiv 1$ lies in $\mathcal{X}$,
each $Z\in\mathcal{Z}$ is also a supermartingale, and we have the
inclusion
\begin{equation}
\mathcal{S} \supseteq \mathcal{Z}.
\label{eq:inclusion}
\end{equation}
The set $\mathcal{Z}$ is convex, and contains the density processes of
equivalent local martingale measures (ELMMs) in situations where such
measures would exist. A feature of our approach is that we shall not
be using any constructions involving ELMMs, even restricted to a
finite horizon, as we discuss further below in Section
\ref{subsubsec:completion}.

The standing no-arbitrage assumption we shall make is that the set of
supermartingale deflators is non-empty:
\begin{equation}
\mathcal{S}(y)\neq \emptyset.  
\label{eq:noarb}
\end{equation}
The condition \eqref{eq:noarb} is equivalent to the condition of no
unbounded profit with bounded risk (NUPBR) (also referred to as no
arbitrage of the first kind, $\mathrm{NA}_{1}$), weaker than the no
free lunch with vanishing risk (NFLVR) condition, the latter being
equivalent to the existence of equivalent local martingale measures
(ELMMs), as established by Delbaen and Schachermayer \cite{ds94} for
the case of a locally bounded semimartingale stock price
process. There are various characterisations of NUPBR, including that
the set $\mathcal{Z}$ of LMDs is non-empty: see Karatzas and Kardaras
\cite{kk07}, Kardaras \cite{k12}, Takaoka and Schweizer \cite{ts14}
and Chau et al \cite{chauetal17}, as well as the recent overview by
Kabanov, Kardaras and Song \cite{kks16}.

\subsubsection{Completion of the stochastic basis and equivalent
measures}
\label{subsubsec:completion}

As indicated earlier, we shall not use equivalent local martingale
measures (ELMMs), even restricted to a finite horizon. This is partly
for aesthetic reasons: since we work under NUPBR and assume only the
existence of various classes of deflators, which is the minimal
requirement for well posed utility maximisation problems, it is
natural to seek proofs which use only deflators.

There is also a mathematical rationale for avoiding ELMMs.  We are
working on an infinite horizon and have have assumed the usual
conditions. Thus, each element of the filtration
$\mathbb{F}=(\mathcal{F}_{t})_{t\geq 0}$ includes all the
$\mathbb{P}$-null sets of
$\mathcal{F}:=\sigma(\bigcup_{t\geq
  0}\mathcal{F}_{t})=:\mathcal{F}_{\infty}$, the tail
$\sigma$-algebra. So, ultimate events (as time $t\uparrow\infty$) of
$\mathbb{P}$-measure zero are included in any finite time
$\sigma$-field $\mathcal{F}_{T},\,T<\infty$.

It is well-known that in such a scenario many financial models will
not admit an equivalent martingale measure over the infinite horizon,
because the candidate change of measure density is not a uniformly
integrable martingale. (This is true of the Black-Scholes model, see
Karatzas and Shreve \cite[Section 1.7]{ks98}.) One then has to proceed
with caution when invoking arguments which utilise equivalent
measures, by finding a consistent way to eliminate the tail
$\sigma$-algebra from the picture when restricting to a finite horizon
$T<\infty$.

One possible way forward is to not complete the space. This route was
taken by Huang and Pag\`es \cite{hp92} in an infinite horizon
consumption model in a complete Brownian market. This is sound, though
care is needed to ensure that no results are used which require the
usual hypotheses to hold.

Another way to proceed, if one wishes to consider equivalent measures
restricted to a finite horizon $T<\infty$, is to augment the space
with null events of a $\sigma$-field generated over a finite horizon
at least as big as $T$, that is by
$\sigma\left(\bigcup_{0\leq t\leq T^{\prime}}\mathcal{F}_{t}\right)$, for some
$0\leq T\leq T^{\prime}<\infty$. This can be done in a consistent way,
and relies on an application of Carath\'eodory's extension theorem
(Rogers and Williams \cite[Theorem II.5.1]{rwvol1}). One can then
obtain equivalent measures in an infinite horizon model when
restricting such measures to any finite horizon. This procedure is
carried out in a Brownian filtration in Karatzas and Shreve
\cite[Section 1.7]{ks98}, with a cautionary example \cite[Example
1.7.6]{ks98}, showing that augmenting the $\sigma$-field generated by
Brownian motion over any finite horizon with null sets of the
corresponding tail $\sigma$-algebra would render invalid the
construction of equivalent measures, even over a finite horizon.

The message is that one has to be careful in using any constructions
involving equivalent measures, even restricted to a finite horizon,
when working in an infinite horizon financial model.

We avoid any such pitfalls, since we avoid ELMMs entirely. In
particular, in Section \ref{sec:bpr} we establish bipolarity results
between the primal and dual domains using only the Stricker and Yan
\cite{sy98} version of the optional decomposition theorem, relying on
deflators rather than martingale measures.

We mention this issue because many papers appear to use a complete
stochastic basis on an infinite horizon, and at the same time then use
equivalent measures over a finite or infinite horizon, without any
statement about the elimination of the tail $\sigma$-field. This
applies to some proofs in papers tackling the infinite horizon
consumption problem (see Mostovyi \cite[Lemma 4.2]{most15} and Chau et
al \cite[Lemma 1]{chauetal17}). In a similar vein, some celebrated
papers working on an infinite horizon, such as the seminal connection
between ELMMs and NFLVR of Delbaen and Schachermayer \cite{ds94}, and
the optional decomposition result of Kramkov \cite{k96}, invoke ELMMs
over an infinite horizon, without seeming to address the issue that
these will not exist over a perpetual timeframe in even the simplest
Brownian model such as the Black-Scholes model, and that care must
sometimes be taken to eliminate the tail $\sigma$-algebra if invoking
ELMMs (even restricted to a finite horizon) in an infinite horizon
model.

We would suggest that it was taken as implicit in the papers cited
above that, when necessary, the tail $\sigma$-algebra was eliminated
in a consistent way when invoking arguments involving ELMMs. But it
should be said that no such qualifying statements were made. We
conjecture that all the arguments in these and other papers where such
potential inconsistencies may arise can be rendered sound by
amendments as described above. This is an issue for possible future
investigation, though fortunately not one we need to address, as we
bypass all these problems by arguments which avoid the use of ELMMs
entirely.

\subsection{The primal problem}
\label{subsec:primal}

Let $U:[0,\infty)\to\mathbb{R}$ denote the agent's utility function,
assumed to be strictly increasing, strictly concave, continuously
differentiable and satisfying the Inada conditions
\begin{equation}
U^{\prime}(0) := \lim_{x\downarrow 0}U^{\prime}(x) = +\infty, \quad
U^{\prime}(\infty) := \lim_{x\to\infty}U^{\prime}(x) = 0. 
\label{eq:inada}
\end{equation}

Let $\kappa:[0,\infty)\to\mathbb{R}_{+}$ be a non-negative,
non-decreasing c\`adl\`ag adapted process, which will act as a finite
measure that will discount utility from inter-temporal
wealth. We assume that $\kappa$ satisfies
\begin{equation*}
\kappa_{0}=0, \quad \mathbb{P}[\kappa_{\infty}>0]>0, \quad
\kappa_{\infty}\leq K,
\end{equation*}
for some finite constant $K$, so that 
$\mathbb{E}\left[\int_{0}^{\infty}\ud\kappa_{t}\right]$ is bounded.

The agent's primal problem is to maximise expected utility from
inter-temporal wealth over the infinite horizon. The primal value
function $u(\cdot)$ is defined by
\begin{equation}
u(x) := \sup_{X\in\mathcal{X}(x)}\mathbb{E}\left[
\int_{0}^{\infty}U(X_{t})\ud\kappa_{t}\right], \quad x>0. 
\label{eq:primal}
\end{equation}
To exclude a trivial problem, we shall assume throughout that the
primal value function satisfies
\begin{equation*}
u(x) > -\infty, \quad \forall\, x>0.
\end{equation*}
This is a mild condition, which can be guaranteed by assuming that for
all wealth processes $X\in\mathcal{X}(x)$ we have
$\mathbb{E}\left[\int_{0}^{\infty}
\min(0,U(X_{t}))\ud\kappa_{t}\right]>-\infty$.

\begin{remark}[Discounted units]
\label{rem:du}

There is no loss of generality in working with discounted quantities
(so in effect a zero interest rate). To see this, suppose instead that
we have a positive interest rate process $r=(r_{t})_{t\geq 0}$, so the
cash asset with initial value $1$ has positive price process
$A_{t}=\e^{\int_{0}^{t}r_{s}\ud s},\,t\geq 0$. If $\widetilde{X}$ is
the un-discounted wealth process, then the problem in
\eqref{eq:primal} is
$\mathbb{E}\left[
  \int_{0}^{\infty}U\left(\widetilde{X}_{t}/A_{t}\right)\ud\kappa_{t}\right]
\to \max!$ We can define another utility function
$\widetilde{U}:\mathbb{R}^{2}_{+}\to\mathbb{R}$ such that
$\widetilde{U}(A_{t},\widetilde{X}_{t})=U(\widetilde{X}_{t}/A_{t}),\,t\geq
0$, and the problem in \eqref{eq:primal} can then be transported to
one in terms of the raw (un-discounted) wealth process. For example,
if $U(\cdot)=\log(\cdot)$ is logarithmic utility, we choose
$\widetilde{U}(A,\widetilde{X})=\log(\widetilde{X})-\log(A)$. If
$U(x)=x^{p}/p,\,p<1,p\neq 0$ is power utility, then we choose
$\widetilde{U}(A,\widetilde{X})=A^{-p}\widetilde{X}^{p}/p$.

\end{remark}

\begin{remark}[Stochastic utility]
\label{rem:su}

In the problem \eqref{eq:primal} we can allow $U(\cdot)$ to be
stochastic, so to also depend on $\omega\in\Omega$ in an optional
way. The analysis is unaffected, as the reader can easily verify, so
one can read the proofs with a stochastic utility in mind and with
dependence on $\omega\in\Omega$ suppressed throughout.
  
\end{remark}

\subsubsection{Some examples}
\label{subsubsec:sexamples}

We list here some examples of inter-temporal wealth utility
maximisation problems, illustrating how the measure $\kappa$ manifests
itself in various cases. Further examples can be found in Bouchard and
Pham \cite[Section 2]{bp04}.

\begin{example}[Perpetual wealth-path-dependent utility maximisation]
\label{examp:pwpdum}

Take $\ud\kappa_{t}=\e^{-\alpha t}\ud t$ for some positive discount rate
$\alpha>0$, and an infinite horizon, so the objective is
\begin{equation}
\mathbb{E}\left[\int_{0}^{\infty}\exp\left(-\alpha t\right)U(X_{t})\ud
  t\right] \to \max !
\label{eq:centralproblem}
\end{equation}
This is the quintessential example we have in mind as our central
problem, and can be thought of as an objective of a long-lived
investment fund building up wealth. We shall treat this example under
power and logarithmic utility in Section \ref{sec:examples}, to
illustrate the application of the duality theorem of the paper, with
different market environments: an incomplete market with a stock with
whose market price of risk is a three-dimensional Bessel process (so
will satisfy NUPBR but not NFLVR) and which has a stochastic
volatility, and a Black-Scholes (thus, complete) market.

There are no esoteric ingredients in \eqref{eq:centralproblem} such as
a random termination time which generates the wealth-path-dependent
objective, but such modifications can be added. Indeed, suppose we
have a random horizon given by $T\sim\Exp(\lambda)$, an exponentially
distributed time with parameter $\lambda>0$, independent of the stock
price filtration. The objective can be re-cast with an additional
integral over the probability density function of $T$, so we have
\begin{equation*}
\mathbb{E}\left[\int_{0}^{T}\exp\left(-\alpha t\right)U(X_{t})\ud
  t\right] = \mathbb{E}\left[\int_{0}^{\infty}\lambda\exp(-\lambda
  s)\int_{0}^{s}\exp(-\alpha t)U(X_{t})\ud t\ud s\right].
\end{equation*}
Integration by parts allows the objective to be re-written as
\begin{equation*}
\mathbb{E}\left[\int_{0}^{\infty}\exp\left(-(\alpha + \lambda)
    t\right)U(X_{t})\ud t\right] \to \max !  
\end{equation*}
so we recover a problem of the same type as in
\eqref{eq:centralproblem} with a modified discount factor.

\end{example}

\begin{example}[Utility of terminal wealth at a random horizon]
\label{examp:utwrxsh}

The other classical example which yields an inter-temporal wealth
objective is where we maximise expected utility of terminal wealth
$\mathbb{E}[U(X_{T})]$ at some random horizon $T$, an almost surely
finite $\mathcal{F}$-measurable non-negative random variable.

For instance, let $T\sim\Exp(\lambda)$ be an exponentially distributed
random time with parameter $\lambda>0$, independent of the stock price
filtration. As in Example \ref{examp:pwpdum}, we re-write the
objective with an integral over the probability density function of
$T$, so we have
\begin{equation*}
\mathbb{E}\left[U(X_{T})\right] =
\mathbb{E}\left[\int_{0}^{\infty}\lambda\exp(-\lambda t)U(X_{t})\ud
t\right],   
\end{equation*}
which again yields a problem of the type in Example
\ref{examp:pwpdum}. We observe that $\kappa$ is given by
\begin{equation*}
\kappa_{t} = \int_{0}^{t}\lambda\exp(-\lambda s)\ud s =
1 - \exp(-\lambda t) = \mathbb{P}[T\leq t], \quad t\geq 0.  
\end{equation*}
The obvious generalisation is to a general random time $T$ which
is independent of the asset price filtration. In this case one has
$\kappa_{t}=\mathbb{P}[T\leq t],\,t\geq 0$ in the inter-temporal
wealth problem \eqref{eq:primal}.

In the case where $T$ is a stopping time we have
$\kappa_{t}=\mathbbm{1}_{\{t\geq T\}},\,t\geq 0$, and this includes
the case where $T$ is deterministic, so there is no time horizon
uncertainty, and we revert to the classical terminal wealth problem.

\end{example}

Further similar examples are given in Bouchard and Pham \cite[Examples
1--3]{bp04}, adapted to the case of a finite horizon for the overall
problem in \eqref{eq:primal}.

\subsection{The budget constraint}
\label{subsec:bc}

Our approach to establishing the form of the dual to the primal
utility maximisation problem \eqref{eq:primal} is to determine an
appropriate supermartingale constraint satisfied by the pair
$(X_{t},(X_{s})_{0\leq s\leq t})_{t\geq 0}$, that is, the value of an
admissible wealth process at any time as well as the wealth path up to
that point. This gives an infinite horizon budget constraint on the
wealth path. Using a supermartingale constraint in this way is
analogous to the procedure followed in consumption problems, where one
considers the wealth process at any time as well the consumption plan
up to that time.

Let $\mathcal{B}$ denote the set of all non-negative c\`adl\`ag
adapted processes $\beta$ satisfying
$\int_{0}^{t}\beta_{s}\ud\kappa_{s}<\infty$ almost surely for all $t\geq
0$:
\begin{equation}
\mathcal{B} := \left\{\beta\geq 0: \mbox{c\`adl\`ag,
adapted, such that $\int_{0}^{\cdot}\beta_{s}\ud\kappa_{s}<\infty$
almost surely}\right\}.   
\label{eq:mcB}
\end{equation}
The processes in $\mathcal{B}$ will act as an additional dual control,
above and beyond that implied in the classical supermartingale
deflators, as we shall see in due course.

For any $\beta\in\mathcal{B}$ and for any supermartingale deflator
$S\in\mathcal{S}(y)$, define a process $R$ by
\begin{equation}
R_{t} := \exp\left(-\int_{0}^{t}\beta_{s}\ud\kappa_{s}\right)S_{t}, \quad
t\geq 0, \quad \beta\in\mathcal{B},\,S\in\mathcal{S}(y), \quad y>0.  
\label{eq:R}
\end{equation}
Denote the set of such processes with initial value $y>0$ by
$\mathcal{R}(y)$:
\begin{equation}
\mathcal{R}(y) :=\left\{R: \mbox{$R$ is defined by
\eqref{eq:R}}\right\}, \quad y>0.  
\label{eq:mcR}
\end{equation}
We write $\mathcal{R}\equiv\mathcal{R}(1)$ and we have
$\mathcal{R}(y)=y\mathcal{R}$ for $y>0$.  We shall prove in Section
\ref{subsec:cdd} that the set $\mathcal{R}$ is convex (see Lemma
\ref{lem:cry}), which will lead to the corresponding property for the
dual domain to the primal problem \eqref{eq:primal}, to be defined
shortly.

Since $\beta\in\mathcal{B} $ is almost surely non-negative, the
supermartingale property of the deflated wealth $SX$ in \eqref{eq:smd}
also holds for $RX$, for any $R\in\mathcal{R}(y)$, so we have the
inclusion $\mathcal{S}(y) \supseteq \mathcal{R}(y),\,y>0$, and each
$R\in\mathcal{R}(y)$ is also a supermartingale.

For each $R\in\mathcal{R}(y)$ and for the same $\beta\in\mathcal{B}$
appearing in the definition \eqref{eq:R}, define a process $Y$ by
\begin{equation}
Y_{t} := \beta_{t}R_{t} =
\beta_{t}\exp\left(-\int_{0}^{t}\beta_{s}\ud\kappa_{s}\right)S_{t},
\quad t\geq 0, \quad \beta\in\mathcal{B},\,S\in\mathcal{S}(y), \quad
y>0.
\label{eq:Y}
\end{equation}
Denote the set of such processes by $\mathcal{Y}(y)$:
\begin{equation}
\mathcal{Y}(y) :=\left\{Y: \mbox{$Y$ is defined by
\eqref{eq:Y}}\right\}, \quad y>0.    
\label{eq:mcY}
\end{equation}
The set $\mathcal{Y}(y)$ will form the domain of the dual problem to
the inter-temporal wealth problem \eqref{eq:primal}, as we shall see
shortly. We shall refer to processes $Y\in\mathcal{Y}(y)$ as
\textit{wealth-path deflators} or inter-temporal wealth deflators or,
simply, as deflators, when no confusion can arise. We write
$\mathcal{Y}\equiv\mathcal{Y}(1)$, with $\mathcal{Y}(y)=y\mathcal{Y}$
for $y>0$. The set $\mathcal{Y}$ turns out to be convex, as we shall
show in Section \ref{subsec:cdd}. This is an important ingredient in
our approach to establishing certain bipolarity relations between the
primal and dual domains, which underlie the duality results of the
paper. As we shall see, the convexity of $\mathcal{Y}$ will stem from
the particular structure of the dual variables as given in
\eqref{eq:Y}. This structure seems to have eluded some previous
studies of inter-temporal wealth utility maximisation problems. We
shall say more on this structure and make comparisons with the
approach of Bouchard and Pham \cite{bp04} in Section
\ref{subsec:oaeb}, after we prove the bipolarity relations.

The following lemma gives the supermartingale constraint and the
resultant infinite horizon budget constraint on admissible wealth
processes, which will lead to the form of the dual problem.

\begin{lemma}[Supermartingale and budget constraints]
\label{lem:sbc}

Let $\beta\in\mathcal{B}$ be any non-negative c\`adl\`ag adapted
process satisfying $\int_{0}^{t}\beta_{s}\ud\kappa_{s}<\infty$ almost
surely for all $t\geq 0$. Define the processes $R\in\mathcal{R}(y)$
and the wealth-path deflators $Y\in\mathcal{Y}(y)$ by \eqref{eq:R} and
\eqref{eq:Y}, respectively. We then have that
\begin{equation}
M := RX + \int_{0}^{\cdot}X_{s}Y_{s}\ud\kappa_{s} \quad \mbox{is a 
supermartingale}.
\label{eq:fsc}
\end{equation}
As a consequence, we have the \emph{infinite horizon budget
  constraint}
\begin{equation}
\mathbb{E}\left[\int_{0}^{\infty}X_{t}Y_{t}\ud\kappa_{t}\right] \leq
xy, \quad x,y>0, \quad \forall\,X\in\mathcal{X}(x),\,Y\in\mathcal{Y}(y).  
\label{eq:ihbc}
\end{equation}

\end{lemma}

\begin{proof}

For $x,y>0$ let $X\in\mathcal{X}(x)$ be an admissible wealth
process and let $S\in\mathcal{S}(y)$ be any supermartingale
deflator. The It\^o product rule applied to
$XR=XS\exp\left(-\int_{0}^{\cdot}\beta_{s}\ud\kappa_{s}\right)$ gives
\begin{equation}
M_{t}:= X_{t}R_{t} +
\int_{0}^{t}X_{s}Y_{s}\ud\kappa_{s} = xy +
\int_{0}^{t}\exp\left(-\int_{0}^{s}\beta_{u}\ud\kappa_{u}\right)
\ud(X_{s}S_{s}), \quad t\geq 0,
\label{eq:xr}
\end{equation}
where we have used the definition \eqref{eq:Y} of the inter-temporal
wealth deflators. Since $XS$ is a supermartingale, it has a Doob-Meyer
decomposition $XS=xy+L-A$ for some local martingale $L$ and a
non-decreasing process $A$, with $L_{0}=A_{0}=0$. Using this
Doob-Meyer decomposition, the integral on the right-hand-side of
\eqref{eq:xr} is also seen to be a supermartingale, so we obtain the
supermartingale property of
$M:=XR + \int_{0}^{\cdot}X_{s}Y_{s}\ud\kappa_{s}$ as stated in the
lemma.

The supermartingale property gives 
\begin{equation*}
\mathbb{E}\left[X_{t}R_{t}
+ \int_{0}^{t}X_{s}Y_{s}\ud\kappa_{s}\right]\leq xy, \quad t\geq 0.
\end{equation*}
Since $XR$ is non-negative, we thus also have
\begin{equation*}
\mathbb{E}\left[\int_{0}^{t}X_{s}Y_{s}\ud\kappa_{s}\right]\leq xy,
\quad t\geq 0. 
\end{equation*}
Letting $t\uparrow\infty$ and using monotone convergence we obtain the
infinite horizon budget constraint \eqref{eq:ihbc}.

\end{proof}

Note that for $\beta\equiv 0$ the supermartingale property in
\eqref{eq:fsc} is simply the statement that $XS$ is a supermartingale
for all $X\in\mathcal{X}$ and $S\in\mathcal{S}$. This is the basic
sense in which we are extending the starting point of the methodology
of Kramkov and Schachermayer \cite{ks99,ks03} towards duality: begin
with a supermartingale constraint to build a budget constraint. The
presence of the supermartingales $S\in\mathcal{S},R\in\mathcal{R}$ in
these arguments will ultimately be exploited to invoke supermartingale
convergence results involving Fatou convergence of processes, in
proving that an abstract dual domain $\mathcal{D}$ (an enlargement of
the domain $\mathcal{Y}$ to encompass processes dominated by some
$Y\in\mathcal{Y}$) is closed (see Lemmata \ref{lem:Dclosed} and
\ref{lem:fatouconvex}).

\subsection{The dual problem}
\label{subsec:dual}

Let $V:\mathbb{R}_{+}\to\mathbb{R}$ denote the convex conjugate of the
utility function, defined by
\begin{equation*}
V(y) := \sup_{x>0}[U(x)-xy], \quad y>0.  
\end{equation*}
The map $y\mapsto V(y),\,y>0$, is strictly convex, strictly decreasing,
continuously differentiable on $\mathbb{R}_{+}$, $-V(\cdot)$ satisfies
the Inada conditions, and we have the bi-dual relation
\begin{equation*}
U(x) := \inf_{y>0}[V(y)+xy], \quad x>0,
\end{equation*}
as well as
$V^{\prime}(\cdot)=-I(\cdot)=-(U^{\prime})^{-1}(\cdot)$,
where $I(\cdot)$ denotes the inverse of marginal utility. In
particular, we have the inequality
\begin{equation}
V(y) \geq U(x)-xy, \quad \forall\, x,y>0, \quad \mbox{with equality
iff $U^{\prime}(x)=y$}.  
\label{eq:VUbound}
\end{equation}

From the budget constraint \eqref{eq:ihbc} we can motivate the form of
the dual problem to \eqref{eq:primal} by bounding the achievable
utility in the familiar way.  For any $X\in\mathcal{X}(x)$ and
$Y\in\mathcal{Y}(y)$ we have
\begin{eqnarray}
\mathbb{E}\left[\int_{0}^{\infty}U(X_{t})\ud\kappa_{t}\right] & \leq & 
\mathbb{E}\left[\int_{0}^{\infty}U(X_{t})\ud\kappa_{t}\right] +
xy - \mathbb{E}\left[\int_{0}^{\infty}X_{t}Y_{t}\ud\kappa_{t}\right]
\nonumber \\
& = & \mathbb{E}\left[\int_{0}^{\infty}\left(U(X_{t}) -
X_{t}Y_{t}\right)\ud\kappa_{t}\right] + xy \nonumber \\ 
& \leq & \mathbb{E}\left[\int_{0}^{\infty}
V\left(Y_{t}\right)\ud\kappa_{t}\right] + xy, \quad x,y>0,  
\label{eq:vubound1}
\end{eqnarray}
the last inequality a consequence of \eqref{eq:VUbound}. This
motivates the definition of the dual problem associated with the
primal problem \eqref{eq:primal}, with dual value function
$v:\mathbb{R}_{+}\to\mathbb{R}$ defined by
\begin{equation}
v(y) := \inf_{Y\in\mathcal{Y}(y)}\mathbb{E}\left[
\int_{0}^{\infty}V(Y_{t})\ud\kappa_{t}\right], \quad y>0.
\label{eq:dual}
\end{equation}
We shall assume that the dual problem is finitely valued:
\begin{equation}
v(y) < \infty, \quad \mbox{for all $y>0$}.  
\label{eq:vfinite}
\end{equation}

\begin{remark}[Reasonable asymptotic elasticity]
\label{rem:aeu}
  
As is known from Kramkov and Schachermayer \cite{ks03},
\eqref{eq:vfinite} is a mild condition that will guarantee a
well-posed primal problem. It is also well known that one can
alternatively impose the reasonable asymptotic elasticity condition of
Kramkov and Schachermayer \cite{ks99} on the utility function:
\begin{equation}
\mathrm{AE}(U) :=
\limsup_{x\to\infty}\frac{xU^{\prime}(x)}{U(x)} < 1,
\label{eq:AEU}
\end{equation}
along with the assumption that $u(x)<\infty$ for some $x>0$. Then, as
in Kramkov and Schachermayer \cite[Note 2]{ks03} or Bouchard and Pham
\cite[Remark 5.1]{bp04}, these conditions can be shown to yield
\eqref{eq:vfinite}.

\end{remark}

\section{The main duality}
\label{sec:tmd}

Here is the central duality statement of the paper.

\begin{theorem}[Perpetual inter-temporal wealth duality under NUPBR]
\label{thm:itwd}

Define the primal inter-temporal wealth utility maximisation problem
by \eqref{eq:primal} and the corresponding dual problem by
\eqref{eq:dual}. Assume \eqref{eq:noarb}, \eqref{eq:inada} and that
\begin{equation*}
u(x)>-\infty,\,\forall \,x>0, \quad v(y)<\infty,\,\forall \,y>0.   
\end{equation*}
Then:

\begin{itemize}

\item[(i)] $u(\cdot)$ and $v(\cdot)$ are conjugate:
\begin{equation*}
v(y) = \sup_{x>0}[u(x)-xy], \quad u(x) =
\inf_{y>0}[v(y)+xy], \quad x,y>0.
\end{equation*}

\item[(ii)] The primal and dual optimisers
$\widehat{X}(x)\in\mathcal{X}(x)$ and
$\widehat{Y}(y)\in\mathcal{Y}(y)$ exist and are unique, so that
\begin{equation*}
u(x) =
\mathbb{E}\left[\int_{0}^{\infty}U(\widehat{X}_{t}(x))\ud\kappa_{t}\right],
\quad v(y) = \mathbb{E}\left[\int_{0}^{\infty}
V(\widehat{Y}_{t}(y))\ud\kappa_{t}\right], \quad x,y>0,
\end{equation*}
with
$\widehat{Y}(y)=\widehat{\beta}\widehat{R}(y)=\widehat{\beta}
\exp\left(-\int_{0}^{\cdot}\widehat{\beta}_{s}\ud\kappa_{s}\right)
\widehat{S}(y)$, for an optimal $\widehat{\beta}\in\mathcal{B}$ and
optimal supermartingales $\widehat{R}(y)\in\mathcal{R}(y)$ and
$\widehat{S}(y)\in\mathcal{S}(y)$.

\item[(iii)] With $y=u^{\prime}(x)$ (equivalently,
$x=-v^{\prime}(y)$), the primal and dual optimisers are related
by
\begin{equation*}
U^{\prime}(\widehat{X}_{t}(x)) = \widehat{Y}_{t}(y), \quad
\mbox{equivalently}, \quad \widehat{X}_{t}(x) =
-V^{\prime}(\widehat{Y}_{t}(y)), \quad t\geq 0,
\end{equation*}
and satisfy
\begin{equation}
\mathbb{E}\left[\int_{0}^{\infty}
\widehat{X}_{t}(x)\widehat{Y}_{t}(y)\ud\kappa_{t}\right] = xy.  
\label{eq:oihbc}
\end{equation}
Moreover, the associated optimal wealth process $\widehat{X}(x)$
satisfies
\begin{equation}
\widehat{X}_{t}(x)\widehat{R}_{t}(y)  = \mathbb{E}\left[\left.
\int_{t}^{\infty}\widehat{X}_{s}(x)\widehat{Y}_{s}(y)\ud\kappa_{s}\right
\vert\mathcal{F}_{t}\right], \quad t\geq 0, 
\label{eq:owp}
\end{equation}
and the process
$\widehat{M}:=\widehat{X}(x)\widehat{R}(y) +
\int_{0}^{\cdot}\widehat{X}_{s}(x)\widehat{Y}_{s}(y)\ud\kappa_{s}$ is
a uniformly integrable martingale.

\item[(iv)] The functions $u(\cdot)$ and $-v(\cdot)$ are strictly
increasing, strictly concave, satisfy the Inada conditions, and for
all $x,y>0$ their derivatives satisfy
\begin{equation*}
xu^{\prime}(x) = \mathbb{E}\left[\int_{0}^{\infty}
U^{\prime}(\widehat{X}_{t}(x))\widehat{X}_{t}(x)\ud\kappa_{t}\right],
\quad yv^{\prime}(y) = \mathbb{E}\left[\int_{0}^{\infty}
V^{\prime}(\widehat{Y}_{t}(y))\widehat{Y}_{t}(y)\ud\kappa_{t}\right].
\end{equation*}

\end{itemize}

\end{theorem}

\begin{remark}[The finite horizon case]
\label{rem:fh}

As the analysis in the sequel will show, it is easy to verify that all
our methodology works without alteration for the finite horizon
version of \eqref{eq:primal}, with some terminal time $T<\infty$. The
budget constraint is altered to have an upper limit of $T$ as are all
the results of Theorem \ref{thm:itwd}. We thus extend the problem
studied in Bouchard and Pham \cite{bp04} to the NUPBR scenario, in
addition to the strengthening of the basic duality statement as
described below, where we do not have to assume \textit{a priori} that
the dual domain is the polar of the primal domain.

\end{remark}

The proof of Theorem \ref{thm:itwd} will be given in Section
\ref{sec:pdt}, and will rely on bipolarity results and an abstract
version of the duality theorem in Section \ref{sec:abpd}, with the
bipolarity results proven in Section \ref{sec:bpr}. Duality results
akin to items (i)--(iii) of the theorem (but not the additional novel
characterisation \eqref{eq:owp} of the optimal wealth process) were
obtained by Bouchard and Pham \cite{bp04} over a finite horizon and
under NFLVR. Compared to \cite{bp04}, Theorem \ref{thm:itwd} makes a
stronger statement in other ways. We describe this strengthening
briefly here, and will give further details in Section
\ref{subsec:oaeb} after we prove bipolarity relations between the
primal and dual domains, as some of the features are directly
concerned with such polarity results.

First, we strengthen the duality for inter-temporal wealth utility
maximisation to the weaker no-arbitrage assumption of NUPBR, compared
to the NFLVR assumption in \cite{bp04}. Second, we avoid having to
define the dual domain as the polar of the primal domain. As indicated
in the Introduction, the dual domain in \cite{bp04} was directly
defined as the set of deflators for which a finite horizon version of
the budget constraint holds. In the language of the polar of a set
(defined in Section \ref{sec:abpd}, see Definition \ref{def:sscs}) the
dual space is set equal to the polar of the primal space, by
definition. This automatically renders the dual domain convex and
closed, but the statement of the duality result is then somewhat
weaker, because one half of the perfect bipolarity between the primal
and dual domains (as given in Proposition \ref{prop:abp}) has been
achieved by definition.

In our approach, the dual space arises from the budget constraint
\eqref{eq:ihbc}, itself derived from the supermartingale property
\eqref{eq:fsc} of the process $M$. This renders the budget constraint
a necessary condition for admissibility. On enlarging the primal
domain to include processes dominated by some admissible wealth, we
show in Lemma \ref{lem:suffX} that the budget constraint is also a
sufficient condition for admissibility. This uses the Stricker and Yan
\cite{sy98} version of the optional decomposition theorem, avoiding
martingale measures in favour of local martingale deflators. This
equivalence between primal admissibility and the budget constraint
establishes that the enlarged primal set $\mathcal{C}$ is the polar of
the dual space $\mathcal{Y}$.

We then show that our dual space is convex, relying on the particular
structure of the wealth path deflators in \eqref{eq:Y}. An enlargement
of the dual domain (in a similar vein to the primal enlargement),
combined with supermartingale convergence results which exploit Fatou
convergence of processes, culminates in Lemma \ref{lem:Dclosed}, which
shows that the enlarged dual domain $\mathcal{D}$ is closed (in an
appropriate topology). This, along with convexity and solidity, yields
that the enlarged dual domain $\mathcal{D}$ is the bipolar of the
original domain $\mathcal{Y}$. Thus gives us the perfect bipolarity we
need between $\mathcal{C}$ and $\mathcal{D}$.

The above procedure is in essence the Kramkov and Schachermayer
\cite{ks99,ks03} program for bipolarity and duality, adapted to an
inter-temporal wealth framework. We shall describe these features of
the bipolarity derivations in more detail in Section
\ref{subsec:oaeb}, and compare the program to that of Bouchard and
Pham \cite{bp04}, after we have proven the bipolarity relations.

\begin{remark}[The case where $\kappa\ll\Leb$]
\label{rem:sc}

Theorem \ref{thm:itwd} holds true regardless of whether 
the measure $\kappa$ admits a density with respect to Lebesgue
measure. However, when $\kappa\ll\Leb$ there is a natural change of
variable which one would use in computations, as we shall see in the
course of some examples in Section \ref{sec:examples}, so we highlight
here how the Theorem \ref{thm:itwd} is slightly re-cast in that
case. The scenario to keep in mind is the case where
$\ud\kappa_{t}=\e^{-\alpha t}\ud t$ for a positive impatience rate
$\alpha>0$.

In the definition \eqref{eq:mcB} of the set $\mathcal{B}$, one
replaces $\kappa$ by Lebesgue measure. With an abuse of notation, to
use the same symbol for this set of auxiliary dual controls,
$\mathcal{B}$ now denotes the set of non-negative c\`adl\`ag processes
$\beta$ such that $\int_{0}^{\cdot}\beta_{s}\ud s<\infty$ almost
surely. With similar abuse of notation, the set $\mathcal{R}(y)$ is
composed of processes
$R:=\exp\left(-\int_{0}^{\cdot}\beta_{s}\ud s\right)S$, for
supermartingale deflators $S\in\mathcal{S}(y)$. The wealth-path
deflators are then given by $Y:=\beta R$, and once again we denote the
set of such processes by $\mathcal{Y}(y)$. The supermartingale
property \eqref{eq:fsc} converts to the statement that the process
$M:=XR+\int_{0}^{\cdot}X_{s}Y_{s}\ud s$ is a supermartingale. The
budget constraint \eqref{eq:ihbc} becomes
$\mathbb{E}\left[\int_{0}^{\infty}X_{t}Y_{t}\ud t\right]\leq xy$.

With this notation, define the positive process
$\gamma=(\gamma_{t})_{t\geq 0}$ as the reciprocal of
$(\ud\kappa_{t}/\ud t)_{t\geq 0}$:
\begin{equation*}
\gamma_{t} := \left(\frac{\ud\kappa_{t}}{\ud t}\right)^{-1}, \quad
t\geq 0.
\end{equation*}
The dual problem then takes the form
\begin{equation}
v(y) := \inf_{Y\in\mathcal{Y}(y)}\mathbb{E}\left[
\int_{0}^{\infty}V(\gamma_{t}Y_{t})\ud\kappa_{t}\right],  \quad y>0,
\label{eq:dualprime}
\end{equation}
as can be confirmed by repeating the computation that led to
\eqref{eq:vubound1} in this altered set-up.

With these changes, items (ii)--(iv) of Theorem \ref{thm:itwd} are
altered to:

\begin{itemize}

\item[(ii)$^{\prime}$] The primal and dual optimisers
$\widehat{X}(x)\in\mathcal{X}(x)$ and
$\widehat{Y}(y)\in\mathcal{Y}(y)$ exist and are unique, so that
\begin{equation*}
u(x) =
\mathbb{E}\left[\int_{0}^{\infty}U(\widehat{X}_{t}(x))\ud\kappa_{t}\right],
\quad v(y) = \mathbb{E}\left[\int_{0}^{\infty}
V(\gamma_{t}\widehat{Y}_{t}(y))\ud\kappa_{t}\right], \quad x,y>0,
\end{equation*}
with
$\widehat{Y}(y)=\widehat{\beta}\widehat{R}(y)=\widehat{\beta}
\exp\left(-\int_{0}^{\cdot}\widehat{\beta}_{s}\ud s\right)
\widehat{S}(y)$, for an optimal $\widehat{\beta}\in\mathcal{B}$ and
optimal supermartingales $\widehat{R}(y)\in\mathcal{R}(y)$ and
$\widehat{S}(y)\in\mathcal{S}(y)$.

\item[(iii)$^{\prime}$] With $y=u^{\prime}(x)$ (equivalently,
$x=-v^{\prime}(y)$), the primal and dual optimisers are related
by
\begin{equation}
U^{\prime}(\widehat{X}_{t}(x)) = \gamma_{t}\widehat{Y}_{t}(y), \quad
\mbox{equivalently}, \quad \widehat{X}_{t}(x) =
-V^{\prime}(\gamma_{t}\widehat{Y}_{t}(y)), \quad t\geq 0,
\label{eq:pditwprime}
\end{equation}
and satisfy
\begin{equation}
\mathbb{E}\left[\int_{0}^{\infty}
\widehat{X}_{t}(x)\widehat{Y}_{t}(y)\ud t\right] = xy.  
\label{eq:oihbcprime}
\end{equation}
Moreover, the associated optimal wealth process $\widehat{X}(x)$
satisfies
\begin{equation}
\widehat{X}_{t}(x)\widehat{R}_{t}(y)  = \mathbb{E}\left[\left.
\int_{t}^{\infty}\widehat{X}_{s}(x)\widehat{Y}_{s}(y)\ud s\right
\vert\mathcal{F}_{t}\right], \quad t\geq 0, 
\label{eq:owpprime}
\end{equation}
and the process
$\widehat{M}:=\widehat{X}(x)\widehat{R}(y) +
\int_{0}^{\cdot}\widehat{X}_{s}(x)\widehat{Y}_{s}(y)\ud s$ is a
uniformly integrable martingale.

\item[(iv)$^{\prime}$] The functions $u(\cdot)$ and
  $-v(\cdot)$ are strictly increasing, strictly concave, satisfy the
  Inada conditions, and for all $x,y>0$ their derivatives satisfy
\begin{equation*}
xu^{\prime}(x) = \mathbb{E}\left[\int_{0}^{\infty}
U^{\prime}(\widehat{X}_{t}(x))\widehat{X}_{t}(x)\ud\kappa_{t}\right],
\quad yv^{\prime}(y) = \mathbb{E}\left[\int_{0}^{\infty}
V^{\prime}(\gamma_{t}\widehat{Y}_{t}(y))\widehat{Y}_{t}(y)\ud t\right].
\end{equation*}

\end{itemize}

\end{remark}

\section{Abstract bipolarity and duality}
\label{sec:abpd}

In this section we specify a finite measure space which allows us to
write the primal and dual problems in abstract notation, over suitably
enlarged primal and dual domains. We then state the bipolarity
relations between the abstract primal and dual domains in Proposition
\ref{prop:abp}, which forms the basis for the subsequent abstract
duality of Theorem \ref{thm:adt}.

Set
\begin{equation*}
\mathbf{\Omega}:=[0,\infty)\times\Omega.  
\end{equation*}
Let $\mathcal{G}$ denote the optional $\sigma$-algebra on
$\mathbf{\Omega}$, that is, the sub-$\sigma$-algebra of
$\mathcal{B}([0,\infty))\otimes\mathcal{F}$ generated by evanescent
sets and stochastic intervals of the form
$\llbracket T,\infty\llbracket$ for arbitrary stopping times $T$.
Define the measure
\begin{equation}
\mu:=\kappa\times\mathbb{P}  
\label{eq:mu}
\end{equation}
on $(\mathbf{\Omega},\mathcal{G})$. On the resulting finite measure
space $(\mathbf{\Omega},\mathcal{G},\mu)$, denote by $L^{0}_{+}(\mu)$
the space of non-negative $\mu$-measurable functions, corresponding to
non-negative infinite horizon processes.

The primal and dual domains for our optimisation problems
\eqref{eq:primal} and \eqref{eq:dual} are now considered as subsets of
$L^{0}_{+}(\mu)$. The abstract primal and dual domains will be
enlargements of $\mathcal{X}(x)$ and $\mathcal{Y}(y)$ to accommodate
processes dominated by some element of the original domain in
question.

The abstract primal domain is $\mathcal{C}(x)$, defined by
\begin{equation}
\mathcal{C}(x) := \{g\in L^{0}_{+}(\mu): \mbox{$g\leq X,\,\mu$-a.e., for
some $X\in\mathcal{X}(x)$}\}, \quad x>0.  
\label{eq:Cx}
\end{equation}
We write $\mathcal{C}\equiv\mathcal{C}(1)$, with
$\mathcal{C}(x)=x\mathcal{C}$ for $x>0$, and the set $\mathcal{C}$ is
convex. Since $U(\cdot)$ is increasing, the primal value function of
\eqref{eq:primal} is now written in the abstract notation as an
optimisation over $g\in\mathcal{C}(x)$:
\begin{equation}
u(x) := \sup_{g\in\mathcal{C}(x)}\int_{\mathbf{\Omega}}U(g)\ud\mu,
\quad x>0.
\label{eq:vfabs}
\end{equation}

The abstract dual domain is obtained by a similar enlargement of the
original dual domain. Define the set $\mathcal{D}(y)$ by
\begin{equation}
\mathcal{D}(y) := \{h\in L^{0}_{+}(\mu): \mbox{$h\leq Y,\,\mu$-a.e.,
for some $Y\in\mathcal{Y}(y)$}\}, \quad y>0. 
\label{eq:Dy}
\end{equation}
We write $\mathcal{D}\equiv\mathcal{D}(1)$, we have
$\mathcal{D}(y)=y\mathcal{D}$ for $y>0$, and the set $\mathcal{D}$ is
convex, inheriting this property from $\mathcal{Y}$. This is a crucial
feature, and relies on our demonstration of the convexity of
$\mathcal{Y}$ in Section \ref{subsec:cdd} (see Lemma \ref{lem:cry}),
which in turn relies on the inter-temporal wealth deflators having the
particular structure in \eqref{eq:Y}.

With this notation, and since $V(\cdot)$ is
decreasing, the dual problem \eqref{eq:dual} takes the form
\begin{equation}
v(y) := \inf_{h\in\mathcal{D}(y)}\int_{\mathbf{\Omega}}V(h)\ud\mu,
\quad y>0. 
\label{eq:dvfabs}
\end{equation}

\subsection{Abstract bipolarity}
\label{subsec:abp}

The abstract duality theorem relies on the abstract bipolarity result
in Proposition \ref{prop:abp} below which connects the sets
$\mathcal{C}$ and $\mathcal{D}$. The result is of course in the spirit
of Kramkov and Schachermayer \cite[Proposition 3.1]{ks99}.

We shall sometimes employ the notation
\begin{equation}
\langle g,h\rangle := \int_{\mathbf{\Omega}}gh\ud\mu, \quad
g,h\in L^{0}_{+} (\mu).
\label{eq:notation}
\end{equation}

Let us recall the concepts of set solidity and the polar of a set.

\begin{definition}[Solid set, closed set]
\label{def:sscs}
  
A subset $A\subseteq L^{0}_{+}(\mu)$ is called \emph{solid} if $f\in A$
and $0\leq g\leq f,\,\mu$-a.e. implies that $g\in A$. 

A set is \emph{closed in $\mu$-measure}, or simply {\em closed}, if it
is closed with respect to the topology of convergence in measure
$\mu$.

\end{definition}

\begin{definition}[Polar of a set]
  
The {\em polar}, $A^{\circ}$, of a set $A\subseteq L^{0}_{+}(\mu)$,
is defined by
\begin{equation*}
A^{\circ} := \left\{h\in L^{0}_{+}(\mu): \langle g,h\rangle\leq
1,\,\mbox{for each $g\in A$}\right\}. 
\end{equation*}

\end{definition}

For clarity and for later use, we state here the bipolar theorem of
Brannath and Schachermayer \cite[Theorem 1.3]{bs99}, originally proven
in a probability space, and adapted here to the measure space
$(\mathbf{\Omega},\mathcal{G},\mu)$.

\begin{theorem}[Bipolar theorem, Brannath and Schachermayer
\cite{bs99}, Theorem 1.3]
\label{thm:bp}

On the finite measure space $(\mathbf{\Omega},\mathcal{G},\mu)$:

\begin{itemize}

\item[(i)] For a set $A\subseteq L^{0}_{+}(\mu)$, its polar
  $A^{\circ}$ is a closed, convex, solid subset of $L^{0}_{+}(\mu)$.

\item[(ii)] The bipolar $A^{\circ\circ}$, defined by
\begin{equation*}
A^{\circ\circ} := \left\{g\in L^{0}_{+}(\mu): \langle g,h\rangle\leq
1,\,\mbox{for each $h\in A^{\circ}$}\right\},
\end{equation*}
is the smallest closed, convex, solid set in $L^{0}_{+}(\mu)$
containing $A$.
  
\end{itemize}

\end{theorem}

\begin{proposition}[Abstract bipolarity]
\label{prop:abp}

Under the condition \eqref{eq:noarb}, the abstract primal and dual
sets $\mathcal{C}$ and $\mathcal{D}$ satisfy the following properties:

\begin{itemize}

\item[(i)] $\mathcal{C}$ and $\mathcal{D}$ are both closed with
respect to convergence in measure $\mu$, convex and solid;

\item[(ii)] $\mathcal{C}$ and $\mathcal{D}$ satisfy the bipolarity
relations
\begin{eqnarray}
g\in\mathcal{C} & \iff & \langle g,h\rangle \leq 1, \quad \forall\,
h\in\mathcal{D}, \quad \mbox{that is,
$\mathcal{C}=\mathcal{D}^{\circ}$}, \label{eq:bp1}\\
h\in\mathcal{D} & \iff & \langle g,h\rangle \leq 1, \quad \forall\,
g\in\mathcal{C}, \quad \mbox{that is,
$\mathcal{D}=\mathcal{C}^{\circ}$};  \label{eq:bp2}
\end{eqnarray}

\item[(iii)] $\mathcal{C}$ and $\mathcal{D}$ are bounded in
$L^{0}(\mu)$, and $\mathcal{D}$ is also bounded in $L^{1}(\mu)$.

\end{itemize}

\end{proposition}

The proof of Proposition \ref{prop:abp} will be given in Section
\ref{sec:bpr}, where we shall establish that the infinite horizon
budget constraint is also a sufficient condition for admissibility,
once the primal domain is enlarged to accommodate processes dominated
by admissible wealths. This culminates in the full bipolarity
relations once we enlarge dual domain in a similar manner. The
derivations in Section \ref{sec:bpr} are quite distinct from previous
approaches, and are the bedrock of the mathematical results. As
indicated earlier, we shall establish the bipolarity results without
any recourse whatsoever to constructions involving ELMMs, by
exploiting ramifications of the Stricker and Yan \cite{sy98} version
of the optional decomposition theorem.

\subsection{The abstract duality}
\label{subsec:ad}

Armed with the abstract bipolarity in Proposition \ref{prop:abp}, we
have the following abstract version of the convex duality relations
between the primal problem \eqref{eq:vfabs} and its dual
\eqref{eq:dvfabs}. The theorem shows that all the natural tenets of
utility maximisation theory, as established by Kramkov and
Schachermayer \cite{ks99,ks03} in the terminal wealth problem under
NFLVR, extend to the infinite horizon inter-temporal wealth problem
under NUPBR, with weak underlying assumptions on the primal and dual
domains.

\begin{theorem}[Abstract duality theorem]
\label{thm:adt}

Define the primal value function $u(\cdot)$ by \eqref{eq:vfabs} and the
dual value function by  \eqref{eq:dvfabs}. Assume that the utility function
satisfies the Inada conditions \eqref{eq:inada} and that
\begin{equation}
u(x) > -\infty,\,\forall \, x>0, \quad v(y)< \infty,\,\forall \,y>0.   
\label{eq:minimal}
\end{equation}

Then, with Proposition \ref{prop:abp} in place, we have:

\begin{itemize}

\item[(i)] $u(\cdot)$ and $v(\cdot)$ are conjugate:
\begin{equation}
v(y) = \sup_{x>0}[u(x)-xy], \quad u(x) = \inf_{y>0}[v(y)+xy], \quad
x,y>0.
\label{eq:conjugacy}
\end{equation}

\item[(ii)] The primal and dual optimisers
$\widehat{g}(x)\in\mathcal{C}(x)$ and
$\widehat{h}(y)\in\mathcal{D}(y)$ exist and are unique, so that
\begin{equation*}
u(x) = \int_{\mathbf{\Omega}}U(\widehat{g}(x))\ud\mu, \quad v(y) =
\int_{\mathbf{\Omega}}V(\widehat{h}(y))\ud\mu, \quad x,y>0.  
\end{equation*}

\item[(iii)] With $y=u^{\prime}(x)$ (equivalently,
$x=-v^{\prime}(y)$), the primal and dual optimisers are related by
\begin{equation*}
U^{\prime}(\widehat{g}(x)) = \widehat{h}(y), \quad
\mbox{equivalently}, \quad \widehat{g}(x) =
-V^{\prime}(\widehat{h}(y)), 
\end{equation*}
and satisfy
\begin{equation*}
\langle \widehat{g}(x),\widehat{h}(y)\rangle = xy.
\end{equation*}

\item[(iv)] $u(\cdot)$ and $-v(\cdot)$ are strictly increasing,
strictly concave, satisfy the Inada conditions, and their
derivatives satisfy
\begin{equation*}
xu^{\prime}(x) = \int_{\mathbf{\Omega}}
U^{\prime}(\widehat{g}(x))\widehat{g}(x)\ud\mu, \quad
yv^{\prime}(y) = \int_{\mathbf{\Omega}}
V^{\prime}(\widehat{h}(y))\widehat{h}(y)\ud\mu, \quad
x,y>0. 
\end{equation*}

\end{itemize}

\end{theorem}

The proof of Theorem \ref{thm:adt} will be given in Section
\ref{sec:pdt}, and uses as its starting point the bipolarity result in
Proposition \ref{prop:abp}. 

The duality proof itself follows some of the classical steps (with
adaptations) of Kramkov and Schachermayer \cite{ks99,ks03}. For
completeness and clarity we shall give a full, self-contained
treatment.

\section{Bipolarity relations}
\label{sec:bpr}

In this section we prove Proposition \ref{prop:abp}, which establishes
in particular the bipolarity relations \eqref{eq:bp1} and
\eqref{eq:bp2} between the enlarged primal and dual domains
$\mathcal{C}$ and $\mathcal{D}$ in \eqref{eq:Cx} and \eqref{eq:Dy}.

\subsection{Sufficiency of the budget constraint}
\label{subsec:sbc}

The budget constraint \eqref{eq:ihbc}, as derived in Lemma
\ref{lem:sbc}, constitutes a necessary condition for admissible
inter-temporal wealth processes. Setting $x=y=1$ in \eqref{eq:ihbc},
we thus have the implications
\begin{equation}
X\in\mathcal{X} \implies
\mathbb{E}\left[\int_{0}^{\infty}X_{t}Y_{t}\ud\kappa_{t}\right]\leq 1,
\quad \forall\,Y\in\mathcal{Y},
\label{eq:forwardX}
\end{equation}
and
\begin{equation}
Y\in\mathcal{Y} \implies
\mathbb{E}\left[\int_{0}^{\infty}X_{t}Y_{t}\ud\kappa_{t}\right]\leq 1,
\quad \forall\,X\in\mathcal{X}.
\label{eq:forwardY}
\end{equation}
We wish to establish the reverse implications in some form, if need be
by enlarging the primal and dual domains.

Recall the enlarged primal domain $\mathcal{C}\equiv\mathcal{C}(1)$ in
\eqref{eq:Cx} of processes dominated by admissible wealths with
initial capital $1$. The budget constraint \eqref{eq:ihbc} clearly holds with
$g\in\mathcal{C}$ in place of $X\in\mathcal{X}$, so the implication
\eqref{eq:forwardX} extends from $\mathcal{X}$ to $\mathcal{C}$:
\begin{equation}
g\in\mathcal{C} \implies \mathbb{E}\left[
\int_{0}^{\infty}g_{t}Y_{t}\ud\kappa_{t}\right]\leq 1, \quad
\forall\,Y\in\mathcal{Y}.
\label{eq:forwardC}
\end{equation}
We establish the reverse implication to \eqref{eq:forwardC} in Lemma
\ref{lem:suffX} below. This requires some version of the Optional
Decomposition Theorem (ODT), originally formulated by El Karoui and
Quenez \cite{ekq95} in a Brownian setting. This was generalised to
markets with locally bounded semimartingale stock prices by Kramkov
\cite{k96}, extended to the non-locally bounded case by F\"ollmer and
Kabanov \cite{fkab98}, and to models with constraints by F\"ollmer and
Kramkov \cite{fk97}. The relevant version of the ODT for us is the one
due to Stricker and Yan \cite{sy98}, which uses local martingale
deflators, rather then ELMMs.  We shall use a result from \cite{sy98}
which applies to the super-hedging of American claims, so is designed
to construct a process which can super-replicate a payoff at an
arbitrary time. The salient observation is that this result can also
be used to dominate a process over all times, and this is how we shall
employ it.

For clarity we state here the ODT results we need, and specify
afterwards precisely which results from \cite{sy98} we have taken.

For $t\geq 0$, let $\mathcal{T}(t)$ denote the set of
$\mathbb{F}$-stopping times with values in $[t,\infty)$. For $t=0$,
write $\mathcal{T}\equiv\mathcal{T}(0)$, and recall the set
$\mathcal{Z}$ of local martingale deflators in \eqref{eq:mcZ}.

\begin{theorem}[Stricker and Yan \cite{sy98} ODT]
\label{thm:syodt}

\begin{itemize}
  
\item[(i)] Let $W$ be an adapted non-negative process. The process
$ZW$ is a supermartingale for each $Z\in\mathcal{Z}$ if and only if
$W$ admits a decomposition of the form
\begin{equation*}
W = W_{0} + (\phi\cdot P) - A,
\end{equation*}
where $\phi$ is a predictable $P$-integrable process such that
$Z(\phi\cdot P)$ is a local martingale for each $Z\in\mathcal{Z}$, $A$
is an adapted increasing process with $A_{0}=0$, and for all
$Z\in\mathcal{Z}$ and $T\in\mathcal{T}$,
$\mathbb{E}[Z_{T}A_{T}]<\infty$. In this case, moreover, we have
$\sup_{Z\in\mathcal{Z},T\in\mathcal{T}}\mathbb{E}[Z_{T}A_{T}]\leq
W_{0}$.

\item[(ii)] Let $b=(b_{t})_{t\geq 0}$ be a non-negative c\`adl\`ag
process such that
$\sup_{Z\in\mathcal{Z},T\in\mathcal{T}}\mathbb{E}[Z_{T}b_{T}]<\infty$. Then
there exists an adapted c\`adl\`ag process $W$ that dominates $b$:
$W_{t}\geq b_{t}$ almost surely for all $t\geq 0$, $ZW$ is a
supermartingale for each $Z\in\mathcal{Z}$, and the smallest such
process $W$ is given by
\begin{equation}
W_{t} = \esssup_{Z\in\mathcal{Z},T\in\mathcal{T}(t)}\frac{1}{Z_{t}}
\mathbb{E}[Z_{T}b_{T}\vert\mathcal{F}_{t}], \quad t\geq 0. 
\label{eq:syam}
\end{equation}

\end{itemize}

\end{theorem}

Part (i) of Theorem \ref{thm:syodt} is taken from \cite[Theorem
2.1]{sy98}. Part (ii) is a combination of \cite[Lemma 2.4 and Remark
2]{sy98}. 

The following lemma establishes the reverse implication to
\eqref{eq:forwardC}.

\begin{lemma}
\label{lem:suffX}

Suppose $g$ is a non-negative c\`adl\`ag process satisfying
\begin{equation}
\mathbb{E}\left[\int_{0}^{\infty}g_{t}Y_{t}\ud\kappa_{t}\right]\leq 1,
\quad \forall\,Y\in\mathcal{Y}. 
\label{eq:ihbc1}
\end{equation}
Then, $g\in\mathcal{C}$.

\end{lemma}

\begin{proof}

Since $g$ is assumed to satisfy \eqref{eq:ihbc1} for all
$Y\in\mathcal{Y}$ and because we have the inclusion
\eqref{eq:inclusion}, we see that \eqref{eq:ihbc1} is satisfied for
$Y=\beta\exp\left(-\int_{0}^{\cdot}\beta_{s}\ud\kappa_{s}\right)Z$,
for any non-negative c\`adl\`ag $\beta\in\mathcal{B}$ and for any
local martingale deflator $Z\in\mathcal{Z}$.
  
Fix a stopping time $T\in\mathcal{T}$, and for each $n\in\mathbb{N}$
choose $\beta$ according to
\begin{equation*}
\beta_{t} = \frac{\mathbbm{1}_{\{T\leq
t<T+1/n\}}}{\kappa_{T+1/n}-\kappa_{T}}, \quad t\geq 0, \quad
n\in\mathbb{N}.   
\end{equation*}
Define the process $\nu^{(n)}$ by
\begin{eqnarray*}
\nu^{(n)}_{t} & := & \beta_{t}\exp\left
(-C(\kappa_{T+1/n}-\kappa_{T})\int_{0}^{t}\beta_{s}\ud\kappa_{s}\right) \\
& = & \frac{\mathbbm{1}_{\{T\leq t<T+1/n\}}}{\kappa_{T+1/n}-\kappa_{T}}
\exp\left(-C\int_{0}^{t}\mathbbm{1}_{\{T\leq
s<T+1/n\}}\ud\kappa_{s}\right), \quad t\geq 0,
\end{eqnarray*}
for a constant $C>0$ large enough to ensure that
$C\int_{0}^{t}\mathbbm{1}_{\{T\leq
  s<T+1/n\}}\ud\kappa_{s}\geq\int_{0}^{t}\beta_{s}\ud\kappa_{s},\,t\geq
0$, so that
$\nu^{(n)}\leq\beta\exp\left(-\int_{0}^{\cdot}\beta_{s}\ud\kappa_{s}\right)$
almost surely. We then have, for each $n\in\mathbb N$ and
$Z\in\mathcal{Z}$,
\begin{eqnarray*}
1 & \geq &
\mathbb{E}\left[\int_{0}^{\infty}g_{t}Y_{t}\ud\kappa_{t}\right] \\
& \geq &
\mathbb{E}\left[\int_{0}^{\infty}g_{t}\nu^{(n)}_{t}Z_{t}\ud\kappa_{t}\right] \\
& = & \mathbb{E}\left[\frac{1}{\kappa_{T+1/n}-\kappa_{T}}
\int_{T}^{T+1/n}g_{t}Z_{t}\exp\left(
C(\kappa_{t}-\kappa_{T})\right)\ud\kappa_{t}\right]. 
\end{eqnarray*}
Letting $n\to\infty$, and using Fatou's lemma and the right-continuity
of $Zg$, we obtain
\begin{equation*}
\mathbb E\left[Z_{T}g_{T}\right] \leq 1, \quad
\forall\,Z\in\mathcal{Z},\,T\in\mathcal{T}.   
\end{equation*}
Since $Z\in\mathcal{Z}$ and $T\in\mathcal{T}$ were arbitrary, we have
\begin{equation*}
\sup_{Z\in\mathcal{Z},T\in\mathcal{T}}\mathbb{E}\left[Z_{T}g_{T}\right]
\leq 1 <\infty. 
\end{equation*}
Thus, from part (ii) of the Stricker-Yan version of optional
decomposition, Theorem \ref{thm:syodt}, there exists a c\`adl\`ag
process $W$ that dominates $g$, so
$W_{t}\geq g_{t},\,\mathrm{a.s.},\,\forall t\geq 0$, and $ZW$ is a
supermartingale for each $Z\in\mathcal{Z}$. From \eqref{eq:syam}, the
smallest such $W$ given by
\begin{equation*}
W_{t} = \esssup_{Z\in\mathcal{Z},T\in\mathcal{T}(t)}\frac{1}{Z_{t}}
\mathbb{E}[Z_{T}g_{T}\vert\mathcal{F}_{t}], \quad t\geq 0,
\end{equation*}
so that $W_{0}\leq 1$. Further, by part (i) of Theorem
\eqref{thm:syodt}, there exists a predictable $P$-integrable process
$H$ and an adapted increasing process $A$, with $A_{0}=0$, such that
$W$ has decomposition $W=W_{0}+(H\cdot P)-A$, with $Z(H\cdot P)$ a
local martingale for each $Z\in\mathcal{Z}$, and
$\mathbb{E}[Z_{T}A_{T}]<\infty$ for all $Z\in\mathcal{Z}$ and
$T\in\mathcal{T}$.

Since $W$ dominates $g$, we can define a process $X$ by
\begin{equation*}
X_{t} := 1 + (H\cdot P)_{t}, \quad t\geq 0,
\end{equation*}
which also dominates $g$, since its initial value is no smaller than
$W_{0}$ and we have dispensed with the increasing process $A$. We
observe that $X$ corresponds to the value of a self-financing wealth
process with initial capital $1$ which dominates $g$, so that
$g\in\mathcal{C}$.
  
\end{proof}

We can now assemble consequences of the budget constraint and of Lemma
\ref{lem:suffX} which, combined with the bipolar theorem, gives the
following polarity properties of the set $\mathcal{C}$.

\begin{lemma}[Polarity properties of $\mathcal{C}$]
\label{lem:Cprop}

The set $\mathcal{C}\equiv\mathcal{C}(1)$ of admissible wealth
processes with initial capital $x=1$ is a closed, convex and solid
subset of $L^{0}_{+}(\mu)$. It is equal to the polar of the set
$\mathcal{Y}\equiv\mathcal{Y}(1)$ of \eqref{eq:mcY} with respect to
measure $\mu$:
\begin{equation}
\mathcal{C} = \mathcal{Y}^{\circ},
\label{eq:CY0}
\end{equation}
so that
\begin{equation}
\mathcal{C}^{\circ} = \mathcal{Y}^{\circ\circ},
\label{eq:CY00}
\end{equation}
and $\mathcal{C}$ is equal to its bipolar:
\begin{equation}
\mathcal{C}^{\circ\circ} = \mathcal{C}.
\label{eq:Cbipolar}
\end{equation}

\end{lemma}

\begin{proof}

Lemma \ref{lem:suffX}, combined with the implication in
\eqref{eq:forwardC}, gives the equivalence
\begin{equation*}
g\in\mathcal{C} \iff
\mathbb{E}\left[\int_{0}^{\infty}g_{t}Y_{t}\ud\kappa_{t}\right]\leq 1, \quad
\forall\,Y\in\mathcal{Y}.  
\end{equation*}
Equivalently, in terms of the measure $\mu$ of \eqref{eq:mu}, we have
\begin{equation}
g\in\mathcal{C} \iff 
\int_{\mathbf{\Omega}}gY\ud\mu \leq 1, \quad \forall\,
Y\in\mathcal{Y}.
\label{eq:Cchar0}
\end{equation}
The characterisation \eqref{eq:Cchar0} is the dual representation of
$\mathcal{C}$:
\begin{equation*}
\mathcal{C} = \left\{g\in L^{0}_{+}(\mu):
\langle g,Y\rangle\leq 1, \quad \mbox{for each
$Y\in\mathcal{Y}$}\right\}.
\end{equation*}
This says that $\mathcal{C}$ is the polar of $\mathcal{Y}$,
establishing \eqref{eq:CY0} and thus \eqref{eq:CY00}.

Part (i) of the bipolar theorem, Theorem \ref{thm:bp}, along with
\eqref{eq:CY0}, imply that $\mathcal{C}$ is a closed, convex and solid
subset of $L^{0}_{+}(\mu)$ (since it is equal to the polar of a set)
as claimed. Part (ii) of Theorem \ref{thm:bp} gives
$\mathcal{C}^{\circ\circ}\supseteq\mathcal{C}$ with
$\mathcal{C}^{\circ\circ}$ the smallest closed, convex, solid set
containing $\mathcal{C}$. But since $\mathcal{C}$ is itself closed,
convex and solid, we have \eqref{eq:Cbipolar}.

\end{proof}

\begin{remark}
\label{rem:Cprop}

There are other ways to obtain the closed, convex and solid properties
of $\mathcal{C}$. First, the equivalence \eqref{eq:Cchar0} along with
Fatou's lemma yields that the set $\mathcal{C}$ is closed with respect
to the topology of convergence in measure $\mu$. To see this, let
$(g^{n})_{n\in\mathbb{N}}$ be a sequence in $\mathcal{C}$ which
converges $\mu$-a.e. to an element $g\in L^{0}_{+}(\mu)$. For
arbitrary $Y\in\mathcal{Y}$ we obtain, via Fatou's lemma and the fact
that $g^{n}\in\mathcal{C}$ for each $n\in\mathbb{N}$,
\begin{equation*}
\int_{\mathbf{\Omega}}gY\ud\mu \leq
\liminf_{n\to\infty}\int_{\mathbf{\Omega}}g^{n}Y\ud\mu \leq 1,
\end{equation*}
so by \eqref{eq:Cchar0}, $g\in\mathcal{C}$, and thus $\mathcal{C}$ is
closed. Further, it is straightforward to establish the convexity of
$\mathcal{C}$ (inherited from the convexity of $\mathcal{X}$) from its
definition. Finally, solidity of $\mathcal{C}$ is also clear: if one
can dominate an element $g\in\mathcal{C}$ with a self-financing wealth
process, then one can also dominate any smaller process with the same
portfolio.
  
\end{remark}

\subsection{Convexity of the dual domain}
\label{subsec:cdd}

We now turn to the dual side of the analysis. The first step is to
establish convexity properties of the sets $\mathcal{R}$ and
$\mathcal{Y}$. Here, the particular structure of the dual variables in
\eqref{eq:R} and \eqref{eq:Y} comes into play.

\begin{lemma}
\label{lem:cry}

The sets $\mathcal{R}$ and $\mathcal{Y}$ of \eqref{eq:mcR} and
\eqref{eq:mcY} are convex.  
  
\end{lemma}

\begin{proof}

Take two elements $S^{1},S^{2}\in\mathcal{S}$ and two elements
$\beta^{1},\beta^{2}\in\mathcal{B}$, and define
$R^{1},R^{2}\in\mathcal{R}$ and $Y^{1},Y^{2}\in\mathcal{Y}$ by
\begin{equation*}
R^{i} :=
\exp\left(-\int_{0}^{\cdot}\beta^{i}_{s}\ud\kappa_{s}\right)S^{i},
\quad Y^{i} := \beta^{i}R^{i}, \quad i=1,2.
\end{equation*}
For two constants $\lambda_{1},\lambda_{2}\geq 0$ such that
$\lambda_{1}+\lambda_{2}=1$, define the convex combinations
\begin{equation*}
\overline{S} := \lambda_{1}S^{1} + \lambda_{2}S^{2}, \quad
\overline{R} := \lambda_{1}R^{1} + \lambda_{2}R^{2}, \quad
\overline{Y} := \lambda_{1}Y^{1} + \lambda_{2}Y^{2}.
\end{equation*}
Observe that $\overline{S}\in\mathcal{S}$ because the set
$\mathcal{S}$ of supermartingale deflators is convex.

Since $\beta^{i}\geq 0,\,i=1,2$ and the set $\mathcal{S}$ is convex,
we have
\begin{equation*}
\overline{R} \leq \lambda_{1}S^{1} + \lambda_{2}S^{2} = \overline{S}
\in \mathcal{S}.  
\end{equation*}
We can therefore define a non-negative process
$\widetilde{\beta}\in\mathcal{B}$ by the relation
\begin{equation}
\overline{R} = \exp\left(-\int_{0}^{\cdot}
\widetilde{\beta}_{s}\ud\kappa_{s}\right)\overline{S}, 
\label{eq:ovlr}
\end{equation}
This shows that $\overline{R}\in\mathcal{R}$, so that $\mathcal{R}$ is
convex, as claimed.

Define a non-negative process $\widehat{\beta}\in\mathcal{B}$ by
\begin{equation}
\overline{Y} = \widehat{\beta}\overline{R}.
\label{eq:ovly1}
\end{equation}
To establish that $\mathcal{Y}$ is convex, we need to show the
existence of a process $\bar{\beta}\in\mathcal{B}$ such that
\begin{equation}
\overline{Y} = \bar{\beta}\exp\left(
-\int_{0}^{\cdot}\bar{\beta}_{s}\ud\kappa_{s}\right)\overline{S}.   
\label{eq:ovly2}
\end{equation}
From \eqref{eq:ovlr}, \eqref{eq:ovly1} and \eqref{eq:ovly2} we thus
require $\bar{\beta}$ to satisfy the relation
\begin{equation*}
\bar{\beta}\exp\left(-\int_{0}^{\cdot}
\bar{\beta}_{s}\ud\kappa_{s}\right) = \widehat{\beta}\exp\left(
-\int_{0}^{\cdot}\widetilde{\beta}_{s}\ud\kappa_{s}\right),   
\end{equation*}
which, given processes $\widetilde{\beta}$ and $\widehat{\beta}$, does
have a unique solution for $\bar{\beta}$, due to the monotonicity of
the exponential function. Thus $\mathcal{Y}$ is convex.

\end{proof}

The next step is to attempt to reach some form of reverse polarity
result to \eqref{eq:CY0}. It is here that the enlargement of the dual
domain from $\mathcal{Y}$ to the set $\mathcal{D}$ of \eqref{eq:Dy}
comes into play.

To see why this enlargement is needed, we first observe that the
implication \eqref{eq:forwardY} extends from $\mathcal{X}$ to
$\mathcal{C}$, so we have
\begin{equation}
Y \in \mathcal{Y} \implies
\langle g,Y\rangle \leq 1, \quad \forall \,
g\in\mathcal{C},
\label{eq:YC}
\end{equation}
which implies that
\begin{equation}
\mathcal{Y} \subseteq \mathcal{C}^{\circ}.
\label{eq:YC0}  
\end{equation}
We do not have the reverse inclusion, because we do not have the
reverse implication to \eqref{eq:YC}, so cannot write a full
bipolarity relation between sets $\mathcal{C}$ and $\mathcal{Y}$. The
enlargement from $\mathcal{Y}$ to the set $\mathcal{D}$ resolves the
issue, yielding the inter-temporal wealth bipolarity of Lemma
\ref{lem:itwbp} below. This procedure, in the spirit of Kramkov and
Schachermayer \cite{ks99}, requires us to establish that the enlarged
domain is closed in an appropriate topology. Here is the relevant
result.

\begin{lemma}
\label{lem:Dclosed}

The enlarged dual domain $\mathcal{D}\equiv\mathcal{D}(1)$ of
\eqref{eq:Dy} is closed with respect to the topology of convergence in
measure $\mu$.
  
\end{lemma}

The proof of Lemma \ref{lem:Dclosed} will be given further
below. First, we use the result of the lemma to establish the
bipolarity result below.

\begin{lemma}[Inter-temporal wealth bipolarity]
\label{lem:itwbp}

Given Lemma \ref{lem:Dclosed}, the set $\mathcal{D}$ is a closed,
convex and solid subset of $L^{0}_{+}(\mu)$, and the the sets
$\mathcal{C}$ and $\mathcal{D}$ satisfy the bipolarity relations
\begin{equation}
\mathcal{C} = \mathcal{D}^{\circ}, \quad \mathcal{D} =
\mathcal{C}^{\circ}.
\label{eq:CD0}
\end{equation}

\end{lemma}

\begin{proof}

For any $h\in\mathcal{D}$ there will exist an element
$Y\in\mathcal{Y}$ such that $h\leq Y,\,\mu$-almost everywhere. Hence,
the implication \eqref{eq:YC} holds true with $\mathcal{D}$ in place
of $\mathcal{Y}$:
\begin{equation*}
h\in\mathcal{D} \implies
\langle g,h\rangle\leq 1, \quad \forall\,g\in\mathcal{C},
\end{equation*}
which yields the analogue of \eqref{eq:YC0}:
\begin{equation}
\mathcal{D} \subseteq \mathcal{C}^{\circ}.
\label{eq:DsubsetC0}  
\end{equation}
Combining \eqref{eq:CY00} and \eqref{eq:DsubsetC0} we have
\begin{equation}
\mathcal{D} \subseteq \mathcal{Y}^{\circ\circ}.
\label{eq:DY001}  
\end{equation}

Part (ii) of the bipolar theorem, Theorem \ref{thm:bp}, says that
$\mathcal{Y}^{\circ\circ} \supseteq \mathcal{Y}$ and that
$\mathcal{Y}^{\circ\circ}$ is the smallest closed, convex, solid set
which contains $\mathcal{Y}$. But $\mathcal{D}$ is also closed, convex
and solid (closed due to Lemma \ref{lem:Dclosed}, convexity following
easily from the convexity of $\mathcal{Y}$, and solidity is obvious),
and by definition $\mathcal{D}\supseteq\mathcal{Y}$, so we also have
\begin{equation}
\mathcal{D} \supseteq \mathcal{Y}^{\circ\circ}.  
\label{eq:DY002}
\end{equation}
Thus, \eqref{eq:DY001} and \eqref{eq:DY002} give
\begin{equation}
\mathcal{D} = \mathcal{Y}^{\circ\circ}.
\label{eq:DY00}
\end{equation}
In other words, in enlarging from $\mathcal{Y}$ to $\mathcal{D}$ we
have succeeded in reaching the bipolar of the former.

Combining \eqref{eq:DY00} and \eqref{eq:CY00} we see that
$\mathcal{D}$ is the polar of $\mathcal{C}$,
\begin{equation}
\mathcal{D} = \mathcal{C}^{\circ}, 
\label{eq:DC0}
\end{equation}
so we have the second equality in \eqref{eq:CD0}. From \eqref{eq:DC0}
we get $\mathcal{D}^{\circ}=\mathcal{C}^{\circ\circ}$ which, combined
with \eqref{eq:Cbipolar}, yields the first equality in \eqref{eq:CD0},
and the proof is complete.

\end{proof}

It remains to prove Lemma \ref{lem:Dclosed}, which we used above. We
recall the concept of Fatou convergence of stochastic processes from
F\"ollmer and Kramkov \cite{fk97}, that will be needed.

\begin{definition}[Fatou convergence]
\label{def:fc}

Let $(Y^{n})_{n\in\mathbb{N}}$ be a sequence of processes on a
stochastic basis
$(\Omega,\mathcal{F}, \mathbb{F}:=(\mathcal{F}_{t})_{t\geq
0},\mathbb{P})$, uniformly bounded from below, and let $\tau$ be a
dense subset of $\mathbb{R}_{+}$. The sequence
$(Y^{n})_{n\in\mathbb{N}}$ is said to be \textit{Fatou convergent on
  $\tau$} to a process $Y$ if
\begin{equation*}
Y_{t} = \limsup_{s\downarrow t,\,s\in\tau}\limsup_{n\to\infty}Y^{n}_{s}
=   \liminf_{s\downarrow t,\,s\in\tau}\liminf_{n\to\infty}Y^{n}_{s},
\quad \mbox{a.s $\forall\,t\geq 0$}.
\end{equation*}
If $\tau=\mathbb{R}_{+}$, the sequence is simply called \textit{Fatou
convergent}. 
  
\end{definition}

The relevant consequence for our purposes is F\"ollmer and Kramkov
\cite[Lemma 5.2]{fk97}, that for a sequence $(S^{n})_{n\in\mathbb{N}}$
of supermartingales, uniformly bounded from below, with
$S^{n}_{0}=0,\,n\in\mathbb{N}$, there is a sequence
$(Y^{n})_{n\in\mathbb{N}}$ of supermartingales, with
$Y^{n}\in\conv(S^{n},S^{n+1},\ldots)$, and a supermartingale $Y$ with
$Y_{0}\leq 0$, such that $(Y^{n})_{n\in\mathbb{N}}$ is Fatou
convergent on a dense subset $\tau$ of $\mathbb{R}_{+}$ to $Y$. Here,
$\conv(S^{n},S^{n+1},\ldots)$ denotes a convex combination
$\sum_{k=n}^{N(n)}\lambda_{k}S^{k}$ for $\lambda_{k}\in[0,1]$ with
$\sum_{k=n}^{N(n)}\lambda_{k}=1$. The requirement that $S^{n}_{0}=0$
is of course no restriction, since for a supermartingale with (say)
$S^{n}_{0}=1$ (as we shall have when we apply these results below for
supermartingales in $\mathcal{Y}$), we can always subtract the initial
value $1$ to reach a process which starts at zero.

To prove Lemma \ref{lem:Dclosed} we shall need the following lemma on
Fatou convergence of convex combinations of elements in
$\mathcal{R},\mathcal{S}$ and, as a consequence, $\mathcal{Y}$. This
result could instead have been developed in the course of proving
Lemma \ref{lem:Dclosed}, but it simplifies the proof of the latter a
great deal to establish it separately.

\begin{lemma}
\label{lem:fatouconvex}

Let $\tau$ be a dense subset of $\mathbb{R}_{+}$. Let
$(\widetilde{R}^{n})_{n\in\mathbb{N}}$ be a sequence in $\mathcal{R}$,
so given by
\begin{equation*}
\widetilde{R}^{n} = \exp\left(-\int_{0}^{\cdot}
\widetilde{\beta}^{n}_{s}\ud\kappa_{s}\right)\widetilde{S}^{n},
\quad n\in\mathbb{N},  
\end{equation*}
for a sequence $(\widetilde{\beta}^{n})_{n\in\mathbb{N}}$ in
$\mathcal{B}$ and a sequence of supermartingale deflators
$(\widetilde{S}^{n})_{n\in\mathbb{N}}$ in $\mathcal{S}$. Then for
each $n\in\mathbb{N}$ there exist convex combinations
$R^{n}\in\conv(\widetilde{R}^{n},\widetilde{R}^{n+1},\ldots)\in\mathcal{R}$,
$S^{n}\in\conv(\widetilde{S}^{n},\widetilde{S}^{n+1},\ldots)\in\mathcal{S}$,
and a process $\beta^{n}\in\mathcal{B}$ such that
\begin{equation}
R^{n} =
\exp\left(-\int_{0}^{\cdot}\beta^{n}_{s}\ud\kappa_{s}\right)S^{n},
\quad n\in\mathbb{N},  
\label{eq:RnSn}
\end{equation}
and such that the sequence $(R^{n})_{n\in\mathbb{N}}$ (respectively,
$(S^{n})_{n\in\mathbb{N}}$) is Fatou convergent on $\tau$ to to a
supermartingale $R\in\mathcal{R}$ (respectively, $S\in\mathcal{S}$),
with
\begin{equation}
R =
\exp\left(-\int_{0}^{\cdot}\beta_{s}\ud\kappa_{s}\right)S,
\label{eq:RS}
\end{equation}
for a process $\beta\in\mathcal{B}$. As a consequence, the sequence of
inter-temporal wealth deflators
$(Y^{n})_{n\in\mathbb{N}}\in \mathcal{Y}$ given by
$Y^{n}=\beta^{n}R^{n}$ is Fatou convergent on $\tau$ to the element
$Y=\beta R\in\mathcal{Y}$.

\end{lemma}

\begin{proof}

Since $\mathcal{R}$ and $\mathcal{S}$ are both convex sets, the
convex combinations $R^{n},S^{n}$ of the lemma lie in
$\mathcal{R},\mathcal{S}$, respectively. Indeed, by similar reasoning
as in the proof of Lemma \ref{lem:cry}, for non-negative constants
$(\lambda_{k})_{k=n}^{N(n)}$ such that
$\sum_{k=n}^{N(n)}\lambda_{k}=1$, we have
\begin{equation}
R^{n} := \sum_{k=n}^{N(n)}\lambda_{k}\widetilde{R}^{k} =
\sum_{k=n}^{N(n)}\lambda_{k}\exp\left(-\int_{0}^{\cdot}
\widetilde{\beta}^{k}_{s}\ud\kappa_{s}\right)\widetilde{S}^{k} \leq
\sum_{k=n}^{N(n)}\lambda_{k}\widetilde{S}^{k} =: S^{n},
\label{eq:RnSnbeta}
\end{equation}
which shows that $R^{n}\leq S^{n}$, implying $R^{n}\in\mathcal{R}$ and
$S^{n}\in\mathcal{S}$, and implying the existence of
$\beta^{n}\in\mathcal{B}$ such that \eqref{eq:RnSn} holds. From
F\"ollmer and Kramkov \cite[Lemma 5.2]{fk97} there exist
supermartingales $R$ and $S$ such that the sequences
$(R^{n})_{n\in\mathbb{N}}$ and $(S^{n})_{n\in\mathbb{N}}$ Fatou
converge on $\tau$ to $R$ and $S$ respectively.

Define a supermartingale sequence
$(\widetilde{V}^{n})_{n\in\mathbb{N}}$ by
$\widetilde{V}^{n}:=X\widetilde{S}^{n}$, for $X\in\mathcal{X}$. Once
again from \cite[Lemma 5.2]{fk97} there exists a sequence
$(V^{n})_{n\in\mathbb{N}}$ of supermartingales with each
$V^{n}\in\conv(\widetilde{V}^{n},\widetilde{V}^{n+1},\ldots)=
X\conv(\widetilde{S}^{n},\widetilde{S}^{n+1},\ldots)$, and a
supermartingale $V$, such that $(V^{n})_{n\in\mathbb{N}}$ is Fatou
convergent on $\tau$ to $V$. Since
$V^{n}\in X\conv(\widetilde{S}^{n},\widetilde{S}^{n+1},\ldots)$ for
each $n\in\mathbb{N}$, we have $V^{n}=XS^{n}$, for
$S^{n}\in\conv(\widetilde{S}^{n},\widetilde{S}^{n+1},\ldots)$. Because
the sequence $(S^{n})_{n\in\mathbb{N}}$ is Fatou convergent on $\tau$
to the supermartingale $S$, the sequence
$(V^{n})_{n\in\mathbb{N}}=(XS^{n})_{n\in\mathbb{N}}$ is Fatou
convergent on $\tau$ to the supermartingale $V=XS$. Since $XS$ is a
supermartingale and $X\in\mathcal{X}$, we have $S\in\mathcal{S}$.

The same argument as in the last paragraph, now applied to the
supermartingale sequence $(\widetilde{W}^{n})_{n\in\mathbb{N}}$
defined by $\widetilde{W}^{n}:=X\widetilde{R}^{n}$, establishes that
$R\in\mathcal{S}$. But because $R^{n}\leq S^{n},\,\mu$-a.e., we have
$R\leq S,\,\mu$-a.e., so there exists a process $\beta\in\mathcal{B}$
such that \eqref{eq:RS} holds, and thus in fact we have
$R\in\mathcal{R}\subseteq\mathcal{S}$. We have thus established that
the sequence in \eqref{eq:RnSn} Fatou converges to the process $R$ in
\eqref{eq:RS}, and this implies that the sequence
$(Y^{n})_{n\in\mathbb{N}}$ defined by $Y^{n}:=\beta^{n}R^{n}$ must
Fatou converge to a process $\beta R=:Y\in\mathcal{Y}$, since the same
process $\beta^{n}\in\mathcal{B}$ appears in the sequence in
\eqref{eq:RnSn} as well as in the sequence $(Y^{n})_{n\in\mathbb{N}}$,
and the proof is complete.
   
\end{proof}

With this preparation, we can now prove Lemma \ref{lem:Dclosed}.
  
\begin{proof}[Proof of Lemma \ref{lem:Dclosed}]
  
Let $(h^{n})_{n\in\mathbb{N}}$ be a sequence in $\mathcal{D}$,
converging $\mu$-a.e. to some $h\in L^{0}_{+}(\mu)$. We want to show
that $h\in\mathcal{D}$.

Since $h^{n}\in\mathcal{D}$, for each $n\in\mathbb{N}$ we have
$h^{n}\leq\widehat{Y}^{n},\,\mu$-a.e for some element
$\widehat{Y}^{n}\in\mathcal{Y}$ given by
$\widehat{Y}^{n}=\widehat{\beta}^{n}\widehat{R}^{n}$, for a
non-negative process $\widehat{\beta}^{n}\in\mathcal{B}$ and a
supermartingale $\widehat{R}^{n}\in\mathcal{R}$ given by
$\widehat{R}^{n}=\exp\left(-\int_{0}^{\cdot}
\widehat{\beta}^{n}_{s}\ud\kappa_{s}\right)\widehat{S}^{n}$, for a
supermartingale deflator $\widehat{S}^{n}\in\mathcal{S}$.

Consider a convex combination
\begin{equation}
Y^{n} = \sum_{k=n}^{N(n)}\lambda_{k}\widehat{Y}^{k} \geq
\sum_{k=n}^{N(n)}\lambda_{k}h^{k}, \quad n\in\mathbb{N},
\label{eq:Ynconvex}
\end{equation}
for non-negative constants $(\lambda_{k})_{k=n}^{N(n)}$ such that
$\sum_{k=n}^{N(n)}\lambda_{k}=1$.

By convexity of the set $\mathcal{Y}$, we have $Y^{n}\in\mathcal{Y}$
for each $n\in\mathbb{N}$, so there exist processes
$\beta^{n}\in\mathcal{B},R^{n}\in\mathcal{R},S^{n}\in\mathcal{S}$
such that
\begin{equation*}
Y^{n} = \beta^{n}R^{n} = \beta^{n}\exp\left(-\int_{0}^{\cdot}
\beta^{n}_{s}\ud\kappa_{s}\right)S^{n}, \quad n\in\mathbb{N}. 
\end{equation*}

By convexity of the sets $\mathcal{R}$ and $\mathcal{S}$ there will exist
sequences $(\widetilde{R}^{n})_{n\in\mathbb{N}}$ in $\mathcal{R}$ and
$(\widetilde{S}^{n})_{n\in\mathbb{N}}$ in $\mathcal{S}$, such that
$R^{n}\in\conv(\widetilde{R}^{n},\widetilde{R}^{n+1},\ldots)\in\mathcal{R}$,
and  
$S^{n}\in\conv(\widetilde{S}^{n},\widetilde{S}^{n+1},\ldots)\in\mathcal{S}$,
and these convex combinations will in general differ from that in
\eqref{eq:Ynconvex}. We thus have the analogue of \eqref{eq:RnSnbeta}:
\begin{equation*}
R^{n} =
\sum_{k=n}^{\widetilde{N}(n)}\widetilde{\lambda}_{k}\widetilde{R}^{k}
= \sum_{k=n}^{\widetilde{N}(n)}\widetilde{\lambda}_{k}\exp\left(
-\int_{0}^{\cdot}\widetilde{\beta}^{k}_{s}\ud\kappa_{s}\right)
\widetilde{S}^{k} \leq
\sum_{k=n}^{\widetilde{N}(n)}\widetilde{\lambda}_{k}\widetilde{S}^{k}
= S^{n}, \quad n\in\mathbb{N}, 
\end{equation*}
for some sequence $(\widetilde{\beta}^{n})_{n\in\mathbb{N}}$ in
$\mathcal{B}$. and non-negative constants
$(\widetilde{\lambda}_{k})_{k=n}^{\widetilde{N}(n)}$ such that
$\sum_{k=n}^{\widetilde{N}(n)}\widetilde{\lambda}_{k}=1$. By Lemma
\ref{lem:fatouconvex}, the sequences $(R^{n})_{n\in\mathbb{N}}$ in
$\mathcal{R}$ and $(S^{n})_{n\in\mathbb{N}}$ in $\mathcal{S}$ Fatou
converge on a dense subset $\tau$ of $\mathbb{R}_{+}$ to
supermartingales $R\in\mathcal{R}$ and $S\in\mathcal{S}$,
respectively, and such that \eqref{eq:RS} holds for some process
$\beta\in\mathcal{B}$. Then, again by Lemma \ref{lem:fatouconvex}, the
sequence $(Y^{n})_{n\in\mathbb{N}}$ Fatou converges on $\tau$ to
$Y=\beta R\in\mathcal{Y}$. So the first term in \eqref{eq:Ynconvex}
converges to $Y\in\mathcal{Y}$ while the last term converges to $h$ as
$n\to\infty$, so the inequality in \eqref{eq:Ynconvex} gives $h\leq Y$,
and thus $h\in\mathcal{D}$.

\end{proof}

With the inter-temporal wealth bipolarity of Lemma \ref{lem:itwbp}, we
can establish Proposition \ref{prop:abp}.

\begin{proof}[Proof of Proposition \ref{prop:abp}]

From the properties of $\mathcal{C}$ established in Lemma
\ref{lem:Cprop}, we have all the claimed properties of $\mathcal{C}$
in items (i) and (ii). The corresponding assertions for $\mathcal{D}$
follow from Lemma \ref{lem:itwbp}.

For item (iii), consider first the set $\mathcal{D}$. Since the wealth
process $X\equiv 1\in\mathcal{X}$, the constant function
$g\equiv 1\in\mathcal{C}$, and the budget constraint (equivalently,
the polar relation \eqref{eq:bp1}) in this case gives
$\int_{\mathbf{\Omega}}h\ud\mu\leq 1$, so $\mathcal{D}$ is bounded in
$L^{1}(\mu)$ and hence in $L^{0}(\mu)$.

For the $L^{0}$-boundedness of $\mathcal{C}$, we shall find a positive
element $\overline{h}\in\mathcal{D}$ and show that $\mathcal{C}$ is
bounded in $L^{1}(\overline{h}\ud\mu)$, and hence bounded in
$L^{0}(\mu)$. Since the constant supermartingale
$S\equiv 1\in\mathcal{S}$ and since the constant process
$\beta\equiv\alpha>0$ for some positive constant $\alpha$, lies in
$\mathcal{B}$, we can take
$\mathcal{Y}\owns
\overline{Y}_{t}:=\alpha\exp(-\alpha\kappa_{t}),\,t\geq 0$, and then
choose $\mathcal{D}\owns\overline{h}\equiv\overline{Y}$. We see that
$\overline{h}\in\mathcal{D}$ is strictly positive except on a set of
$\mu$-measure zero. Then, the budget constraint (equivalently, the
polar relation \eqref{eq:bp2}) gives
$\int_{\mathbf{\Omega}}g\overline{h}\ud\mu\leq 1$ for any
$g\in\mathcal{C}$. Thus, $\mathcal{C}$ is bounded in
$L^{1}(\overline{h}\ud\mu)$ and hence bounded in $L^{0}(\mu)$.

\end{proof}

\subsection{On approaches to establishing bipolarity}
\label{subsec:oaeb}

In this section we compare the approach we have taken to establishing
the polar relations \eqref{eq:bp1} and \eqref{eq:bp2} in Proposition
\ref{prop:abp}, between the enlarged primal and dual domains
$\mathcal{C}$ and $\mathcal{D}$, with the approach taken by Bouchard
and Pham \cite{bp04}. This is instructive and will indicate how we
have been able to strengthen the statement of the final duality
result, in essence by proving, as opposed to partially assuming, the
polar relations, which is what Bouchard and Pham \cite{bp04} were
compelled to do.

\subsubsection{The Kramkov-Schachermayer approach}
\label{subsubsec:ksa}

Our approach is in the spirit of the recipe created by Kramkov and
Schachermayer \cite{ks99,ks03} for the terminal wealth utility
maximisation problem, adapted to an inter-temporal framework. One
begins with a supermartingale property linking the elements of the
primal and dual domains. (In the terminal wealth problem one has the
admissible wealth processes $X\in\mathcal{X}$ and the supermartingale
deflators $S\in\mathcal{S}$, with $XS$ a supermartingale for each
$X\in\mathcal{X}$ and $S\in\mathcal{S}$.)  Here, we invoke the
additional dual controls $\beta\in\mathcal{B}$, and from these and the
supermartingale deflators we construct the supermartingales
$R\in\mathcal{R}$ and the inter-temporal wealth deflators
$Y\in\mathcal{Y}$ according to the relations in \eqref{eq:R} and
\eqref{eq:Y}, repeated below for the case $y=1$, so for
$S\in\mathcal{S}$:
\begin{equation}
R: = \exp\left(-\int_{0}^{\cdot}\beta_{s}\ud\kappa_{s}\right)S, \quad
Y := \beta R, \quad \beta\in\mathcal{B},\,S\in\mathcal{S}.
\label{eq:RY}
\end{equation}
Observe that the deflators $Y\in\mathcal{Y}$ are given by
$Y=\nu S,\,S\in\mathcal{S}$, with the process $\nu$ given by
\begin{equation}
\nu_{t} :=
\beta_{t}\exp\left(-\int_{0}^{t}\beta_{s}\ud\kappa_{s}\right), \quad
t\geq 0, \quad \beta\in\mathcal{B}.
\label{eq:nu}
\end{equation}
We see that $\nu$ satisfies
\begin{equation*}
\int_{0}^{\infty}\nu_{t}\ud\kappa_{t} = 1 -
\exp\left(-\int_{0}^{\infty}\beta_{t}\ud\kappa_{t}\right) \leq 1,
\quad \mbox{almost surely},
\end{equation*}
and hence also
$\mathbb{E}\left[\int_{0}^{\infty}\nu_{t}\ud\kappa_{t}\right] \leq 1$
or, in the notation of \eqref{eq:notation},
\begin{equation}
\langle\nu,1\rangle \leq 1.
\label{eq:nuconstraint1}
\end{equation}
This structure of dual variables for wealth-path-dependent utility
maximisation problems, namely a multiplicative auxiliary control which
augments the classical deflators and which satisfies a constraint of
the form in \eqref{eq:nuconstraint1}, is not uncommon, and we shall
see a similar feature shortly when we describe the Bouchard and Pham
\cite{bp04} approach. The key insight that arises in our approach is
that this auxiliary control must have the very specific structure in
\eqref{eq:nu}, which confers convexity to the dual domain.

From \eqref{eq:RY} and the properties of $S\in\mathcal{S}$, we get
that the process $M$ in \eqref{eq:fsc} is a supermartingale, and in
turn this gives the budget constraint \eqref{eq:ihbc}, repeated below
for the case $x=y=1$, as a necessary condition for admissibility of a
wealth process:
\begin{equation*}
\mathbb{E}\left[\int_{0}^{\infty}X_{t}Y_{t}\ud\kappa_{t}\right] \leq
1, \quad \forall\, X\in\mathcal{X},\, Y\in\mathcal{Y}.  
\end{equation*}
Then, enlarging the primal domain from $\mathcal{X}$ to $\mathcal{C}$,
Lemma \ref{lem:suffX} establishes that the budget constraint is also a
sufficient condition for admissibility, so we obtain the polar
properties of Lemma \ref{lem:Cprop} for $\mathcal{C}$
\begin{equation*}
\mathcal{C} = \mathcal{Y}^{\circ}, \quad \mathcal{C}^{\circ} =
\mathcal{Y}^{\circ\circ}, \quad \mathcal{C}^{\circ\circ} =
\mathcal{C}, 
\end{equation*}
which imply that $\mathcal{C}$ is a closed, convex and solid (CCS)
subset of $L^{0}_{+}(\mu)$. 

Now to the dual side of the story. Using the particular form of the
dual variables in \eqref{eq:RY} we established in Lemma \ref{lem:cry}
that the dual domain $\mathcal{Y}$ is convex. This convexity is passed
on to the enlarged dual domain $\mathcal{D}$. Then, again using the
structure in \eqref{eq:RY}, and in particular that the deflators
$Y\in\mathcal{Y}$ contain the supermartingales
$R\in\mathcal{R},S\in\mathcal{S}$, we are able to exploit Fatou
convergence of supermartingales to show that $\mathcal{D}$ is closed
with respect to the topology of convergence in $\mu$-measure. This,
along with the convexity and (obvious) solidity of $\mathcal{D}$,
shows that $\mathcal{D}$ is also a CCS subset of $L^{0}_{+}(\mu)$,
matching the property we obtained for $\mathcal{C}$. In particular, we
obtain the key result that the enlargement from $\mathcal{Y}$ to
$\mathcal{D}$ has taken as to the bipolar of the original dual domain:
\begin{equation*}
\mathcal{D} = \mathcal{Y}^{\circ\circ}.  
\end{equation*}
This result then readily combines with the earlier polarity properties
of $\mathcal{C}$ to establish the perfect bipolarity relations
\eqref{eq:bp1} and \eqref{eq:bp2}.

The message is that we have made the Kramkov and Schachermayer
\cite{ks99,ks03} prescription for obtaining bipolarity work: begin
with a supermartingale property to arrive at the correct definition of
the dual variables, make no assumptions regarding convexity and closed
properties of either the primal or dual domains, show that with a
natural enlargement of these domains to obtain solid sets, all the
required CCS properties of the domains, and hence bipolarity,
follows. This bipolarity is then the bedrock of the subsequent
program for the proof of the duality theorem, as we shall see in
Section \ref{sec:pdt}.

This methodology is to be contrasted with the approach in \cite{bp04},
which we now describe.

\subsubsection{The Bouchard-Pham approach}
\label{subsubsec:bpa}

The first difference between our methodology and that of Bouchard and
Pham \cite{bp04} is that in \cite{bp04}, the dual domain (let us call
in $\mathcal{D}^{\mathrm{BP}}$) is \textit{defined} as the polar of
the primal domain. Over a finite horizon $T<\infty$, the dual
variables $Y^{\mathrm{BP}}$ and dual domain are thus defined according
to
\begin{equation*}
\mathcal{D}^{\mathrm{BP}} := \left\{Y^{\mathrm{BP}}\geq 0:
\mathbbm{E}\left[\int_{0}^{T}X_{t}Y^{\mathrm{BP}}_{t}\ud\kappa_{t}\right]\leq
1, \forall\,X\in\mathcal{X}\right\},
\end{equation*}
(see the definition of the set $\mathcal{D}(y)$ in \cite[Page
584]{bp04}). In other words, 
\begin{equation}
\mathcal{D}^{\mathrm{BP}} := \mathcal{X}^{0},
\label{eq:BPpolar}
\end{equation}
by assumption. This automatically confers the CCS property to the dual
domain, but the statement of the result is weakened, having been
obtained by definition. The reason that this approach had to be
adopted, we conjecture, is that the authors of \cite{bp04} did not
have to hand the specific structure of the dual variables in
\eqref{eq:RY} that emerges in our approach.

This conjecture is reinforced by the reasoning which now follows. In a
subsequent refinement Bouchard and Pham \cite{bp04} show that, under
an assumption called {\bf (Hf)} (namely, that $\kappa$ decomposes into
a continuous density plus a linear combination of indicator functions
of the form $\mathbbm{1}_{\{\tau\leq t\}},\,t\in[0,T]$, for any
$\mathbb{F}$-stopping time $\tau$), processes of the form
$\nu^{\mathrm{BP}}Z^{\mathcal{M}}$ lie in their dual domain, where
$Z^{\mathcal{M}}$ is the density process of an ELMM, and
$\nu^{\mathrm{BP}}$ is any process satisfying
$\langle\nu^{\mathrm{BP}},1\rangle_{T}:=\mathbb{E}\left[\int_{0}^{T}
  \nu^{\mathrm{BP}}_{t}\ud\kappa_{t}\right]\leq 1$. The similarity
with the structure we have in \eqref{eq:nuconstraint1} is clear. If we
denote the set of processes $\nu^{\mathrm{BP}}Z^{\mathcal{M}}$ by
$\mathcal{Z}^{\mathrm{BP}}$, then under their additional assumption
{\bf (Hf)}, Bouchard and Pham \cite{bp04} are able to re-cast their
dual problem as a minimisation over the convex hull of
$\mathcal{Z}^{\mathrm{BP}}$. This, therefore, is the analogue, under
NFLVR and over a finite horizon, of the dual structure we have used,
but with two caveats. First, they have to use the convex hull of
$\mathcal{Z}^{\mathrm{BP}}$, because the set
$\mathcal{Z}^{\mathrm{BP}}$ is not known to be convex in
general. Second, this lack of convexity is due to the fact that the
authors of \cite{bp04} do not have the particular structure of the
auxiliary dual control $\nu^{\mathrm{BP}}$ that we have found in
\eqref{eq:nu}, a structure that was crucial in our establishing the
convexity of our dual domain. All that is known about the processes
$\nu^{\mathrm{BP}}$ is that they satisfy
$\langle\nu^{\mathrm{BP}},1\rangle_{T}\leq 1$, and this is not enough
to afford a proof of convexity of $\mathcal{Z}^{\mathrm{BP}}$.

Finally, the discussion above also explains why the bulk of the
analysis in \cite{bp04} is carried out on the primal side of the
problem. Since the definition in \eqref{eq:BPpolar} confers the CCS
property to the dual domain by assumption, the remaining work in
\cite{bp04} is concerned with enlarging the primal domain to confer
solidity and proving the remaining polarity relation, as can be
verified by examining \cite[Section 5]{bp04}. 

In summary, we are able to strengthen the duality statement in
\cite{bp04} to any horizon and under NUPBR, by making the broad
pattern of the Kramkov and Schachermayer \cite{ks99,ks03} program for
bipolarity work, without having to assume the associated properties of
either the primal or dual domain. Instead, we begin with a natural
supermartingale property linking the primal and dual elements, thus
identifying the natural dual space for the problem, along with its
particular structure, so that the closed and convex features of the
domains, from which the existence and uniqueness of the optimisers are
ultimately deduced, are demonstrated, as opposed to being assumed.

\section{Proofs of the duality theorems}
\label{sec:pdt}

In this section we prove the abstract duality of Theorem
\ref{thm:adt}, from which the concrete duality of Theorem
\ref{thm:itwd} is then deduced. Throughout this section, we have in
place the result of Proposition \ref{prop:abp}, as this bipolarity is
the starting point of the duality proof. The proof of Theorem
\ref{thm:adt} proceeds via a series of lemmas. The procedure has a
similar flavour to that of Kramkov and Schachermayer \cite{ks99,ks03}
for an abstract duality proof in the context of the terminal wealth
utility maximisation problem, with variations where appropriate, and
with an additional result, Proposition \ref{prop:owp}, which gives the
additional characterisation \eqref{eq:owp} of the optimal wealth
process as well as the uniformly integrable martingale property of the
process
$\widehat{M}:=\widehat{X}(x)\widehat{R}(y) +
\int_{0}^{\cdot}\widehat{X}_{s}(x)\widehat{Y}_{s}(y)\ud\kappa_{s}$. This
proposition also establishes that the process
$\widehat{X}(x)\widehat{R}(y)$ is a potential, and that its limiting
value is $\lim_{t\to\infty}\widehat{X}_{t}(x)\widehat{R}_{t}(y)=0$
almost surely.

Let us state the basic properties that are taken as given throughout
this section.

\begin{fact}
\label{fact:bf}
  
Throughout this section, assume that the utility function satisfies
the Inada conditions \eqref{eq:inada}, that the sets $\mathcal{C}$ and
$\mathcal{D}$ satisfy all the properties in Proposition
\ref{prop:abp}, and that the abstract primal and dual value functions
in \eqref{eq:vfabs} and \eqref{eq:dvfabs} satisfy the minimal
conditions in \eqref{eq:minimal}.

\end{fact}

\textit{All subsequent lemmata and propositions in this section
  implicitly take Fact \ref{fact:bf} as given.}

The first step is to establish weak duality. 

\begin{lemma}[Weak duality]
\label{lem:weakdual}
  
The primal and dual value functions $u(\cdot)$ and $v(\cdot)$ of
\eqref{eq:vfabs} and \eqref{eq:dvfabs} satisfy the weak duality bounds
\begin{equation}
v(y) \geq \sup_{x>0}[u(x)-xy], \quad y>0, \quad
\mbox{equivalently} \quad u(x) \leq \inf_{y>0}[v(y)+xy], \quad
x>0. 
\label{eq:weakd}
\end{equation}
As a result, $u(x)$ is finitely valued for all $x>0$. Moreover, we
have the limiting relations
\begin{equation}
\limsup_{x\to\infty}\frac{u(x)}{x} \leq 0, \quad
\liminf_{y\to\infty}\frac{v(y)}{y} \geq 0. 
\label{eq:limiting}
\end{equation}

\end{lemma}

\begin{proof}

Recall the inequality \eqref{eq:vubound1}. By the same argument
carried out in the measure space $(\mathbf{\Omega},\mathcal{G},\mu)$
we have, for any $g\in\mathcal{C}(x)$ and $h\in\mathcal{D}(y)$, using
the polarity relations in \eqref{eq:bp1} and \eqref{eq:bp2},
\begin{eqnarray}
\int_{\mathbf{\Omega}}U(g)\ud\mu & \leq &
\int_{\mathbf{\Omega}}U(g)\ud\mu + xy -
\int_{\mathbf{\Omega}}gh\ud\mu \nonumber \\
& = & \int_{\mathbf{\Omega}}(U(g)-gh)\ud\mu + xy \nonumber \\ 
& \leq & \int_{\mathbf{\Omega}}V(h)\ud\mu + xy, \quad x,y>0, 
\label{eq:ucbound}         
\end{eqnarray}
the last inequality a consequence of \eqref{eq:VUbound}. Maximising
the left-hand-side of \eqref{eq:ucbound} over $g\in\mathcal{C}(x)$ and
minimising the right-hand-side over $h\in\mathcal{D}(y)$ gives
$u(x)\leq v(y)+xy$ for all $x,y>0$, and \eqref{eq:weakd} follows.

The assumption that $v(y)<\infty$ for all $y>0$ immediately yields
that $u(x)$ is finitely valued for some $x>0$. Since $U(\cdot)$ is
strictly increasing and strictly concave, and given the convexity of
$\mathcal{C}$, these properties are inherited by $u(\cdot)$, which is
therefore finitely valued for all $x>0$. Finally, the relations in
\eqref{eq:weakd} easily lead to those in \eqref{eq:limiting}.
  
\end{proof}

Above, we obtained concavity and monotonicity of $u(\cdot)$ by using
convexity of $\mathcal{C}$ and the properties of $U(\cdot)$. Similar
arguments show that $v(\cdot)$ is strictly decreasing and strictly
convex. We shall see these properties reproduced in proofs of
existence and uniqueness of the optimisers for $u(\cdot),v(\cdot)$.

The next step is to give a compactness lemma for the dual
domain.

\begin{lemma}[Compactness lemma for $\mathcal{D}$]
\label{lem:Dcompact}

Let $(\tilde{h}^{n})_{n\in\mathbb{N}}$ be a sequence in
$\mathcal{D}$. Then there exists a sequence $(h^{n})_{n\in\mathbb{N}}$
with $h^{n}\in\conv(\tilde{h}^{n}, \tilde{h}^{n+1},\ldots)$, which
converges $\mu$-a.e. to an element $h\in\mathcal{D}$ that is
$\mu$-a.e. finite.

\end{lemma}

\begin{proof}

Delbaen and Schachermayer \cite[Lemma A1.1]{ds94} (adapted from a
probability space to the finite measure space
$(\mathbf{\Omega},\mathcal{G},\mu)$) implies the existence of a
sequence $(h^{n})_{n\in\mathbb{N}}$, with
$h^{n}\in\conv(\tilde{h}^{n}, \tilde{h}^{n+1},\ldots)$, which
converges $\mu$-a.e. to an element $h$ that is $\mu$-a.e. finite
because $\mathcal{D}$ is bounded in $L^{0}(\mu)$ (the finiteness also
from \cite[Lemma A1.1]{ds94}). By convexity of $\mathcal{D}$, each
$h^{n},\,n\in\mathbb{N}$ lies in $\mathcal{D}$. Finally, by Fatou's
lemma, for every $g\in\mathcal{C}$ we have
\begin{equation*}
\int_{\mathbf{\Omega}}gh\ud\mu =
\int_{\mathbf{\Omega}}\liminf_{n\to\infty}gh^{n}\ud\mu \leq
\liminf_{n\to\infty}\int_{\mathbf{\Omega}}gh^{n}\ud\mu \leq 1,  
\end{equation*}
so that $h\in\mathcal{D}$.

\end{proof}

Results in the style of Lemma \ref{lem:Dcompact} are standard in these
duality proofs. We will see a similar result for the primal domain
$\mathcal{C}$ shortly.

The next step in the chain of results we need is a uniform
integrability result for the family
$(V^{-}(h))_{h\in\mathcal{D}(y)}$. This will facilitate a proof of
existence and uniqueness of the dual minimiser, and of the conjugacy
for the value functions by establishing the first relation in
\eqref{eq:conjugacy}.

\begin{lemma}[Uniform integrability of
$(V^{-}(h))_{h\in\mathcal{D}(y)}$]
\label{lem:Vminush}

The family $(V^{-}(h))_{h\in\mathcal{D}(y)}$ is uniformly integrable,
for any $y>0$.
  
\end{lemma}

The style of the proof is along identical lines to Kramkov and
Schachermayer \cite[Lemma 3.2]{ks99}, but we give the proof for
completeness.

\begin{proof}[Proof of Lemma \ref{lem:Vminush}]

Since $V(\cdot)$ is decreasing, we need only consider the case where
$V(\infty):=\lim_{y\to\infty}V(y)=-\infty$ (otherwise there is nothing
to prove). Let $\varphi:(-V(0),-V(\infty))\mapsto(0,\infty)$ denote the
inverse of $-V(\cdot)$. Then $\varphi(\cdot)$ is strictly
increasing. For any $h\in\mathcal{D}(y)$ (so
$\int_{\mathbf{\Omega}}h\ud\mu\leq y$) we have, for all $y>0$,
\begin{equation*}
\int_{\mathbf{\Omega}}\varphi(V^{-}(h))\ud\mu \leq \varphi(0) +
\int_{\mathbf{\Omega}}\varphi(-V(h))\ud\mu = \varphi(0) +
\int_{\mathbf{\Omega}}h\ud\mu \leq \varphi(0) + y.
\end{equation*}
Then, using l'H\^opital's rule and the change of variable
$\varphi(x)=y\iff x=-V(y)$, and recalling the function
$I(\cdot)=-V^{\prime}(\cdot)$ (the inverse of marginal utility
$U^{\prime}(\cdot)$), we have
\begin{equation}
\lim_{x\to -V(\infty)}\frac{\varphi(x)}{x} =
\lim_{x\to\infty}\frac{\varphi(x)}{x} =
\lim_{y\to\infty}\frac{y}{-V(y)} =
\lim_{y\to\infty}\frac{1}{I(y)} = +\infty, 
\label{eq:phioverx}
\end{equation}
on using the Inada conditions \eqref{eq:inada}. The
$L^{1}(\mu)$-boundedness of $\mathcal{D}(y)$ means we can apply the de
la Vall\'ee-Poussin theorem (Pham \cite[Theorem A.1.2]{pham09}) which,
combined with \eqref{eq:phioverx}, implies the uniform integrability
of the family $(V^{-}(h))_{h\in\mathcal{D}(y)}$.

\end{proof}

One can can now proceed to prove either existence of a unique
optimiser in the dual problem, or conjugacy of the value functions. We
proceed first with the former, followed by conjugacy.

\begin{lemma}[Dual existence]
\label{lem:dualexis}

The optimal solution $\widehat{h}(y)\in\mathcal{D}(y)$ to the dual
problem \eqref{eq:dvfabs} exists and is unique, so that $v(\cdot)$ is
strictly convex.  
  
\end{lemma}

\begin{proof}

Fix $y>0$. Let $(h^{n})_{n\in\mathbb{N}}$ be a minimising sequence
in $\mathcal{D}(y)$ for $v(y)<\infty$. That is
\begin{equation}
\lim_{n\to\infty}\int_{\mathbf{\Omega}}V(h^{n})\ud\mu = v(y) < \infty.
\label{eq:maxseqdual}
\end{equation}
By the compactness lemma for $\mathcal{D}$ (and thus also for
$\mathcal{D}(y)=y\mathcal{D}$), Lemma \ref{lem:Dcompact}, we can find
a sequence $(\widehat{h}^{n})_{n\in\mathbb{N}}$ of convex
combinations, so
$\mathcal{D}(y)\owns\widehat{h}^{n}\in \conv(h^{n},
h^{n+1},\ldots),\,n\in\mathbb{N}$, which converges $\mu$-a.e. to some
element $\widehat{h}(y)\in\mathcal{D}(y)$. We claim that
$\widehat{h}(y)$ is the dual optimiser. That is, that we have
\begin{equation}
\int_{\mathbf{\Omega}}V(\widehat{h}(y))\ud\mu = v(y).
\label{eq:dualoptimiser}
\end{equation}
From convexity of $V(\cdot)$ and \eqref{eq:maxseqdual} we deduce that
\begin{equation*}
\lim_{n\to\infty}\int_{\mathbf{\Omega}}V(\widehat{h}^{n})\ud\mu \leq
\lim_{n\to\infty}\int_{\mathbf{\Omega}}V(h^{n})\ud\mu = v(y),
\end{equation*}
which, combined with the obvious inequality
$v(y)\leq
\lim_{n\to\infty}\int_{\mathbf{\Omega}}V(\widehat{h}^{n})\ud\mu$ means
that we also have, further to \eqref{eq:maxseqdual},
\begin{equation*}
\lim_{n\to\infty}\int_{\mathbf{\Omega}}V(\widehat{h}^{n})\ud\mu =
v(y).  
\end{equation*}
In other words
\begin{equation}
\lim_{n\to\infty}\int_{\mathbf{\Omega}}V^{+}(\widehat{h}^{n})\ud\mu -
\lim_{n\to\infty}\int_{\mathbf{\Omega}}V^{-}(\widehat{h}^{n})\ud\mu =
v(y) < \infty,    
\label{eq:Vplusminus}
\end{equation}
and note therefore that both integrals in \eqref{eq:Vplusminus} are
finite.

From Fatou's lemma, we have
\begin{equation}
\lim_{n\to\infty}\int_{\mathbf{\Omega}}V^{+}(\widehat{h}^{n})\ud\mu
\geq \int_{\mathbf{\Omega}}V^{+}(\widehat{h}(y))\ud\mu.  
\label{eq:Vplus}
\end{equation}
From Lemma \ref{lem:Vminush} we have uniform integrability of
$(V^{-}(\widehat{h}^{n}))_{n\in\mathbb{N}}$, so that
\begin{equation}
\lim_{n\to\infty}\int_{\mathbf{\Omega}}V^{-}(\widehat{h}^{n})\ud\mu
= \int_{\mathbf{\Omega}}V^{-}(\widehat{h}(y))\ud\mu.  
\label{eq:Vminus}
\end{equation}
Thus, using \eqref{eq:Vplus} and \eqref{eq:Vminus} in
\eqref{eq:Vplusminus}, we obtain
\begin{equation*}
v(y) \geq \int_{\mathbf{\Omega}}V(\widehat{h}(y))\ud\mu,
\end{equation*}
which, combined with the obvious inequality
$v(y)\leq\int_{\mathbf{\Omega}}V(\widehat{h}(y))\ud\mu$, yields
\eqref{eq:dualoptimiser}. The uniqueness of the dual optimiser
follows from the strict convexity of $V(\cdot)$, as does the strict
convexity of $v(\cdot)$. For this last claim, fix $y_{1}<y_{2}$ and
$\lambda\in(0,1)$, note that
$\lambda\widehat{h}(y_{1}) +
(1-\lambda)\widehat{h}(y_{2})\in\mathcal{D}(\lambda y_{1} +
(1-\lambda)y_{2})$ (yet must be sub-optimal for
$v(\lambda y_{1}+(1-\lambda)y_{2})$ as it is not guaranteed to equal
$\widehat{h}(\lambda y_{1}+(1-\lambda)y_{2})$) and therefore, using
the strict convexity of $V(\cdot)$,
\begin{equation*}
v(\lambda y_{1} + (1-\lambda)y_{2})  \leq
\int_{\mathbf{\Omega}}V\left(\lambda\widehat{h}(y_{1}) +
  (1-\lambda)\widehat{h}(y_{2})\right)\ud\mu < \lambda v(y_{1}) +
(1-\lambda)v(y_{2}). 
\end{equation*}

\end{proof}

We now establish conjugacy of the value functions. The method is
similar to the classical method of proof in Kramkov and Schachermayer
\cite[Lemma 3.4]{ks99}, and works by bounding the elements in the
primal domain to create a compact set for the weak$*$ topology
$\sigma(L^{\infty},L^{1})$ on $L^{\infty}(\mu)$,\footnote{Recall that
  a sequence $(g^{n})_{n\in\mathbb{N}}$ in $L^{\infty}(\mu)$ converges
  to $g\in L^{\infty}(\mu)$ with respect to the weak$*$ topology
  $\sigma(L^{\infty},L^{1})$ if and only if
  $(\langle g^{n},h\rangle)_{n\in\mathbb{N}}$ converges to
  $\langle g,h\rangle$ for each $h\in L^{1}(\mu)$.} so as to apply the
minimax theorem, involving a maximisation over a compact set and a
minimisation over a subset of a vector space. This uses the fact that
the dual domain is bounded in $L^{1}(\mu)$.

For the convenience of the reader here is the minimax theorem as we
shall apply it (see Strasser \cite[Theorem 45.8]{strasser85}).

\begin{theorem}[Minimax]
\label{thm:rminimax}

Let $\mathcal{X}$ be a $\sigma(E^{\prime},E)$-compact convex subset of
the topological dual $E^{\prime}$ of a normed vector space $E$, and
let $\mathcal{Y}$ be a convex subset of $E$. Assume that
$f:\mathcal{X}\times\mathcal{Y}\to\mathbb{R}$ satisfies the following
conditions:

\begin{enumerate}

\item $x\mapsto f(x,y)$ is continuous and concave on $\mathcal{X}$ for
  every $y\in\mathcal{Y}$;

\item $y\mapsto f(x,y)$ is convex on $\mathcal{Y}$ for every
  $x\in\mathcal{X}$.
  
\end{enumerate}

Then:
\begin{equation*}
\sup_{x\in\mathcal{X}}\inf_{y\in\mathcal{Y}}f(x,y) =
\inf_{y\in\mathcal{Y}}\sup_{x\in\mathcal{X}}f(x,y). 
\end{equation*}
  
\end{theorem}

Here is the conjugacy result for the primal and dual value functions.

\begin{lemma}[Conjugacy]
\label{lem:conjugacy}
  
The dual value function in \eqref{eq:dvfabs} satisfies the conjugacy
relation
\begin{equation*}
v(y) = \sup_{x>0}[u(x)-xy], \quad \mbox{for each $y>0$},
\end{equation*}
where $u(\cdot)$ is the primal value function in \eqref{eq:vfabs}.

\end{lemma}

\begin{proof}

For $n\in\mathbb{N}$ denote by $\mathcal{B}_{n}$ the set of elements
in $L^{0}_{+}(\mu)$ lying in a ball of radius $n$:
\begin{equation*}
\mathcal{B}_{n} := \left\{g\in L^{0}_{+}(\mu): g\leq
n,\,\mu-\mathrm{a.e.}\right\}. 
\end{equation*}
The sets $(\mathcal{B}_{n})_{n\in\mathbb{N}}$ are
$\sigma(L^{\infty},L^{1})$-compact. Because each $h\in\mathcal{D}(y)$
is $\mu$-integrable, $\mathcal{D}(y)$ is a closed, convex subset of
the vector space $L^{1}(\mu)$, so we apply the minimax theorem as
given in Theorem \ref{thm:rminimax} to the compact set
$\mathcal{B}_{n}$ ($n$ fixed) and the set $\mathcal{D}(y)$, with the
function $f(g,h):=\int_{\mathbf{\Omega}}(U(g)-gh)\ud\mu$, for
$g\in\mathcal{B}_{n},\,h\in\mathcal{D}(y)$, to give
\begin{equation}
\sup_{g\in\mathcal{B}_{n}}\inf_{h\in\mathcal{D}(y)}
\int_{\mathbf{\Omega}}(U(g)-gh)\ud\mu
= \inf_{h\in\mathcal{D}(y)}\sup_{g\in\mathcal{B}_{n}}
\int_{\mathbf{\Omega}}(U(g)-gh)\ud\mu.
\label{eq:mm}
\end{equation}
By the bipolarity relation $\mathcal{C}=\mathcal{D}^{\circ}$ in
\eqref{eq:bp1}, an element $g\in L^{0}_{+}(\mu)$ lies in
$\mathcal{C}(x)$ if and only if
$\sup_{h\in\mathcal{D}(y)}\int_{\mathbf{\Omega}}gh\ud\mu\leq
xy$. Thus, the limit as $n\to\infty$ on the left-hand-side of
\eqref{eq:mm} is given as
\begin{equation}
\label{eq:lhsmm}
\lim_{n\to\infty}\sup_{g\in\mathcal{B}_{n}}\inf_{h\in\mathcal{D}(y)}
\int_{\mathbf{\Omega}}(U(g)-gh)\ud\mu =
\sup_{x>0}\sup_{g\in\mathcal{C}(x)}\left(\int_{\mathbf{\Omega}}U(g)\ud\mu
- xy\right) = \sup_{x>0}[u(x)-xy].
\end{equation}
Now consider the right-hand-side of \eqref{eq:mm}. Define
\begin{equation*}
V_{n}(y):= \sup_{0<x\leq n}[U(x)-xy], \quad y>0, \quad n\in\mathbb{N}.
\end{equation*}
The right-hand-side of \eqref{eq:mm} is then given as
\begin{equation*}
\inf_{h\in\mathcal{D}(y)}\sup_{g\in\mathcal{B}_{n}}
\int_{\mathbf{\Omega}}(U(g)-gh)\ud\mu =
\inf_{h\in\mathcal{D}(y)}\int_{\mathbf{\Omega}}V_{n}(h)\ud\mu =:
v_{n}(y),  
\end{equation*}
so that taking the limit as $n\to\infty$ and equating this with the
limit obtained in \eqref{eq:lhsmm}, we have
\begin{equation}
\lim_{n\to\infty}v_{n}(y) = \sup_{x>0}[u(x)-xy] \leq v(y),
\label{eq:limun}  
\end{equation}
with the inequality due to the weak duality bound in
\eqref{eq:weakd}. Consequently, we will be done if we can now show
that we also have
\begin{equation*}
\lim_{n\to\infty}v_{n}(y) \geq v(y).  
\end{equation*}

Evidently, $(v_{n}(y))_{n\in\mathbb{N}}$ is an increasing sequence
satisfying the limiting inequality in \eqref{eq:limun}. Let
$(\tilde{h}^{n})_{n\in\mathbb{N}}$ be a minimising sequence in
$\mathcal{D}(y)$ for $\lim_{n\to\infty}v_{n}(y)$, so such that
\begin{equation*}
\lim_{n\to\infty}\int_{\mathbf{\Omega}}V_{n}(\tilde{h}^{n})\ud\mu =
\lim_{n\to\infty}v_{n}(y).
\end{equation*}
The compactness lemma for $\mathcal{D}$, Lemma \ref{lem:Dcompact},
implies the existence of a sequence $(h^{n})_{n\in\mathbb{N}}$ in
$\mathcal{D}(y)$, with
$h^{n}\in\conv(\tilde{h}^{n},\tilde{h}^{n+1},\ldots)$, which converges
$\mu$-a.e. to an element $h\in\mathcal{D}(y)$. Now, $V_{n}(y)=V(y)$
for $y\geq I(n)$, where $I(\cdot)=-V^{\prime}(\cdot)$ is the inverse
of $U^{\prime}(\cdot)$ (and $V_{n}(\cdot)\to V(\cdot)$ as
$n\to\infty$). So we deduce from Lemma \ref{lem:Vminush} that
the sequence $(V^{-}_{n}(h^{n}))_{n\in\mathbb{N}}$ is uniformly
integrable, and hence that
\begin{equation}
\lim_{n\to\infty}\int_{\mathbf{\Omega}}V^{-}_{n}(h^{n})\ud\mu =
\int_{\mathbf{\Omega}}V^{-}(h)\ud\mu.  
\label{eq:Vminush1}
\end{equation}
On the other hand, from Fatou's lemma, we have
\begin{equation}
\lim_{n\to\infty}\int_{\mathbf{\Omega}}V^{+}_{n}(h^{n})\ud\mu \geq
\int_{\mathbf{\Omega}}V^{+}(h)\ud\mu,
\label{eq:Vplush1}
\end{equation}
so \eqref{eq:Vminush1} and \eqref{eq:Vplush1} give
\begin{equation}
\lim_{n\to\infty}\int_{\mathbf{\Omega}}V_{n}(h^{n})\ud\mu \geq
\int_{\mathbf{\Omega}}V(h)\ud\mu.
\label{eq:limVh}
\end{equation}
Finally, using convexity of $V_{n}(\cdot)$ and \eqref{eq:limVh}, we obtain
\begin{equation*}
\lim_{n\to\infty}v_{n}(y) =
\lim_{n\to\infty}\int_{\mathbf{\Omega}}V_{n}(\tilde{h}^{n})\ud\mu \geq
\lim_{n\to\infty}\int_{\mathbf{\Omega}}V_{n}(h^{n})\ud\mu \geq
\int_{\mathbf{\Omega}}V(h)\ud\mu \geq v(y),
\end{equation*}
and the proof is complete.

\end{proof}

We now move on to the primal side of the analysis. The first step is
an analogous compactness result to Lemma \ref{lem:Dcompact}, this time
for the primal domain. The proof is identical to the proof of Lemma
\ref{lem:Dcompact} so is omitted.

\begin{lemma}[Compactness lemma for $\mathcal{C}$]
\label{lem:Ccompact}

Let $(\tilde{g}^{n})_{n\in\mathbb{N}}$ be a sequence in
$\mathcal{C}$. Then there exists a sequence $(g^{n})_{n\in\mathbb{N}}$
with $g^{n}\in\conv(\tilde{g}^{n}, \tilde{g}^{n+1},\ldots)$, which
converges $\mu$-a.e. to an element $g\in\mathcal{C}$ that is
$\mu$-a.e. finite.

\end{lemma}

To prove existence of a unique primal optimiser we also need a result
analogous to Lemma \ref{lem:Vminush}, on the uniform integrability of
a sequence $(U^{+}(g^{n}))_{n\in\mathbb{N}}$ for
$g^{n}\in\mathcal{C}(x)$. The proof is in the style of Kramkov and
Schachermayer \cite[Lemma 1]{ks03}.

\begin{lemma}[Uniform integrability of
$(U^{+}(g^{n}))_{n\in\mathbb{N}},\,g^{n}\in\mathcal{C}(x)$] 
\label{lem:Uplusg}

Let $(g^{n})_{n\in\mathbb{N}}$ be a sequence in $\mathcal{C}(x)$, for
any fixed $x>0$. The sequence $(U^{+}(g^{n}))_{n\in\mathbb{N}}$ is
uniformly integrable.
  
\end{lemma}

\begin{proof}

Fix $x>0$. If $U(\infty)\leq 0$ there is nothing to prove, so assume
$U(\infty)>0$.

If the sequence $(U^{+}(g^{n}))_{n\in\mathbb{N}}$ is not uniformly
integrable, then, passing if need be to a subsequence still denoted by 
$(g^{n})_{n\in\mathbb{N}}$, we can find a constant $\alpha>0$ and a
disjoint sequence $(A_{n})_{n\in\mathbb{N}}$ of sets of
$(\mathbf{\Omega},\mathcal{G})$ (so
$A_{n}\in\mathcal{G},\,n\in\mathbb{N}$ and $A_{i}\cap A_{j}=\emptyset$
if $i\neq j$) such that
\begin{equation*}
\int_{\mathbf{\Omega}}U^{+}(g^{n})\mathbbm{1}_{A_{n}}\ud\mu \geq
\alpha, \quad n\in\mathbb{N}.  
\end{equation*}
(See for example Pham \cite[Corollary A.1.1]{pham09}.) Define a
sequence $(f^{n})_{n\in\mathbb{N}}$ of elements in $L^{0}_{+}(\mu)$ by
\begin{equation*}
f^{n} := x_{0} + \sum_{k=1}^{n}g^{k}\mathbbm{1}_{A_{k}},
\end{equation*}
where $x_{0}:=\inf\{x>0:\,U(x)\geq 0\}$.

For any $h\in\mathcal{D}$ (so satisfying
$\int_{\mathbf{\Omega}}h\ud\mu\leq 1$) we have
\begin{equation*}
\int_{\mathbf{\Omega}}f^{n}h\ud\mu =
\int_{\mathbf{\Omega}}\left(x_{0} +
\sum_{k=1}^{n}g^{k}\mathbbm{1}_{A_{k}}\right)h\ud\mu \leq x_{0} +
\sum_{k=1}^{n}\int_{\mathbf{\Omega}}g^{k}h\mathbbm{1}_{A_{k}}\ud\mu
\leq x_{0} + nx.
\end{equation*}
Thus, $f^{n}\in\mathcal{C}(x_{0}+nx),\,n\in\mathbb{N}$.

On the other hand, since $U^{+}(\cdot)$ is non-negative and
non-decreasing,
\begin{eqnarray*}
\int_{\mathbf{\Omega}}U(f^{n})\ud\mu & = &
\int_{\mathbf{\Omega}}U^{+}(f^{n})\ud\mu \\
& = & \int_{\mathbf{\Omega}}U^{+}\left(x_{0} +
\sum_{k=1}^{n}g^{k}\mathbbm{1}_{A_{k}}\right)\ud\mu \\
& \geq & \int_{\mathbf{\Omega}}
U^{+}\left(\sum_{k=1}^{n}g^{k}\mathbbm{1}_{A_{k}}\right)\ud\mu \\
& = & \sum_{k=1}^{n}\int_{\mathbf{\Omega}}
U^{+}\left(g^{k}\mathbbm{1}_{A_{k}}\right)\ud\mu \geq \alpha n.
\end{eqnarray*}
Therefore,
\begin{equation*}
\limsup_{z\to\infty}\frac{u(z)}{z} =
\limsup_{n\to\infty}\frac{u(x_{0}+nx)}{x_{0}+nx} \geq
\limsup_{n\to\infty}\frac{\int_{\mathbf{\Omega}}U(f^{n})\ud\mu}{x_{0}+nx}
\geq \limsup_{n\to\infty}\left(\frac{\alpha n}{x_{0}+nx}\right) =
\frac{\alpha}{x} > 0,
\end{equation*}
which contradicts the limiting weak duality bound in
\eqref{eq:limiting}. This contradiction establishes the result.
  
\end{proof}

One can can now proceed to prove existence of a unique optimiser in
the primal problem. The method of proof is similar to the proof of
dual existence, Lemma \ref{lem:dualexis}, with adjustments for
maximisation as opposed to minimisation and concavity of $U(\cdot)$
replacing convexity of $V(\cdot),$ so is included just for
completeness.

\begin{lemma}[Primal existence]
\label{lem:primexis}

The optimal solution $\widehat{g}(x)\in\mathcal{C}(x)$ to the primal
problem \eqref{eq:vfabs} exists and is unique, so that $u(\cdot)$ is
strictly concave.  
  
\end{lemma}

\begin{proof}

Fix $x>0$. Let $(g^{n})_{n\in\mathbb{N}}$ be a maximising sequence
in $\mathcal{C}(x)$ for $u(x)<\infty$ (the finiteness proven in
Lemma \ref{lem:weakdual}). That is
\begin{equation}
\lim_{n\to\infty}\int_{\mathbf{\Omega}}U(g^{n})\ud\mu = u(x) < \infty.
\label{eq:maxseqprimal}
\end{equation}

By the compactness lemma for $\mathcal{C}$ (and thus also for
$\mathcal{C}(x)=x\mathcal{C}$), Lemma \ref{lem:Ccompact}, we can find
a sequence $(\widehat{g}^{n})_{n\in\mathbb{N}}$ of convex
combinations, so
$\mathcal{C}(x)\owns\widehat{g}^{n}\in \conv(g^{n},
g^{n+1},\ldots),\,n\in\mathbb{N}$, which converges $\mu$-a.e. to some
element $\widehat{g}(x)\in\mathcal{C}(x)$. We claim that
$\widehat{g}(x)$ is the primal optimiser. That is, that we have
\begin{equation}
\int_{\mathbf{\Omega}}U(\widehat{g}(x))\ud\mu = u(x).
\label{eq:primaloptimiser}
\end{equation}
By concavity of $U(\cdot)$ and \eqref{eq:maxseqprimal} we have
\begin{equation*}
\lim_{n\to\infty}\int_{\mathbf{\Omega}}U(\widehat{g}^{n})\ud\mu \geq
\lim_{n\to\infty}\int_{\mathbf{\Omega}}U(g^{n})\ud\mu = u(x),
\end{equation*}
which, combined with the obvious inequality
$u(x)\geq
\lim_{n\to\infty}\int_{\mathbf{\Omega}}U(\widehat{g}^{n})\ud\mu$ means
that we also have, further to \eqref{eq:maxseqprimal},
\begin{equation*}
\lim_{n\to\infty}\int_{\mathbf{\Omega}}U(\widehat{g}^{n})\ud\mu =
u(x).  
\end{equation*}
In other words
\begin{equation}
\lim_{n\to\infty}\int_{\mathbf{\Omega}}U^{+}(\widehat{g}^{n})\ud\mu -
\lim_{n\to\infty}\int_{\mathbf{\Omega}}U^{-}(\widehat{g}^{n})\ud\mu =
u(x) < \infty,    
\label{eq:Uplusminus}
\end{equation}
and note therefore that both integrals in \eqref{eq:Uplusminus} are
finite.

From Fatou's lemma, we have
\begin{equation}
\lim_{n\to\infty}\int_{\mathbf{\Omega}}U^{-}(\widehat{g}^{n})\ud\mu
\geq \int_{\mathbf{\Omega}}U^{-}(\widehat{g}(x))\ud\mu.  
\label{eq:Uminus}
\end{equation}
From Lemma \ref{lem:Uplusg} we have uniform integrability of
$(U^{+}(\widehat{g}^{n}))_{n\in\mathbb{N}}$, so that
\begin{equation}
\lim_{n\to\infty}\int_{\mathbf{\Omega}}U^{+}(\widehat{g}^{n})\ud\mu
= \int_{\mathbf{\Omega}}U^{+}(\widehat{g}(x))\ud\mu.  
\label{eq:Uplus}
\end{equation}
Thus, using \eqref{eq:Uminus} and \eqref{eq:Uplus} in
\eqref{eq:Uplusminus}, we obtain
\begin{equation*}
u(x) \leq \int_{\mathbf{\Omega}}U(\widehat{g}(x))\ud\mu,
\end{equation*}
which, combined with the obvious inequality
$u(x)\geq\int_{\mathbf{\Omega}}U(\widehat{g}(x))\ud\mu$, yields
\eqref{eq:primaloptimiser}. The uniqueness of the primal optimiser
follows from the strict concavity of $U(\cdot)$, as does the strict
concavity of $u(\cdot)$. For this last claim, fix $x_{1}<x_{2}$ and
$\lambda\in(0,1)$, note that
$\lambda\widehat{g}(x_{1}) +
(1-\lambda)\widehat{g}(x_{2})\in\mathcal{C}(\lambda x_{1} +
(1-\lambda)x_{2})$ (yet must be sub-optimal for
$u(\lambda x_{1}+(1-\lambda)x_{2})$ as it is not guaranteed to equal
$\widehat{g}(\lambda x_{1}+(1-\lambda)x_{2})$) and therefore, using
the strict concavity of $U(\cdot)$,
\begin{equation*}
u(\lambda x_{1} + (1-\lambda)x_{2})  \geq
\int_{\mathbf{\Omega}}U\left(\lambda\widehat{g}(x_{1}) +
  (1-\lambda)\widehat{g}(x_{2})\right)\ud\mu > \lambda u(x_{1}) +
(1-\lambda)u(x_{2}). 
\end{equation*}

\end{proof}

We now move on to further characterise the derivatives of the value
functions, as well as the primal and dual optimisers. The first result
is on the derivative of the primal value value function $u(\cdot)$ at
zero (equivalently, the derivative of the dual value function
$v(\cdot)$ at infinity). The proof of the following lemma is in the
style of Kramkov and Schachermayer \cite[Lemma 3.5]{ks99}.

\begin{lemma}
\label{lem:uprime0}

The derivatives of the primal value function in \eqref{eq:vfabs} at
zero and of the dual value function in \eqref{eq:dvfabs} at infinity
are given by
\begin{equation}
u^{\prime}(0) := \lim_{x\downarrow 0}u^{\prime}(x) = +\infty, \quad
-v^{\prime}(\infty) := \lim_{y\to\infty}(-v^{\prime}(y)) = 0.   
\label{eq:uprime0}
\end{equation}

\end{lemma}

\begin{proof}

By the conjugacy result in Lemma \ref{lem:conjugacy} between the value
functions, the assertions in \eqref{eq:uprime0} are equivalent. We
shall prove the second assertion.

The function $-v(\cdot)$ is strictly concave and strictly
increasing, so there is a finite non-negative limit
$-v^{\prime}(\infty):=\lim_{y\to\infty}(-v^{\prime}(y))$. Because
$-V(\cdot)$ is increasing with $\lim_{y\to\infty}(-V^{\prime}(y))=0$, for
any $\epsilon>0$ there exists a number $C_{\epsilon}>0$ such that
$-V(y)\leq C_{\epsilon}+\epsilon y,\,\forall\,y>0$. Using this, the
$L^{1}(\mu)$-boundedness of $\mathcal{D}$ (so that
$\int_{\mathbf{\Omega}}h\ud\mu\leq y,\,\forall\,h\in\mathcal{D}(y)$)
and l'H\^opital's rule, we have, with
$\int_{\mathbf{\Omega}}\ud\mu=:\delta>0$,
\begin{eqnarray*}
0 \leq \lim_{y\to\infty}-v^{\prime}(y) =
\lim_{y\to\infty}\frac{-v(y)}{y} & = &
\lim_{y\to\infty}\sup_{h\in\mathcal{D}(y)}\int_{\mathbf{\Omega}}
\frac{-V(h)}{y}\ud\mu \\
& \leq & \lim_{y\to\infty}\sup_{h\in\mathcal{D}(y)}\int_{\mathbf{\Omega}}
\frac{C_{\epsilon}+\epsilon h}{y}\ud\mu \\
& \leq & \lim_{y\to\infty}\left(\frac{C_{\epsilon}\delta}{y} + \epsilon\right)
= \epsilon,
\end{eqnarray*}
and taking the limit as $\epsilon\downarrow 0$ gives the result.
  
\end{proof}

The final step in the series of lemmas that will furnish us with the
proof of the abstract duality of Theorem \ref{thm:adt} is to
characterise the derivative of the primal value value function
$u(\cdot)$ at infinity (equivalently, the derivative of the dual value
function $v(\cdot)$ at zero) along with a duality characterisation of
the primal and dual optimisers.

\begin{lemma}
\label{lem:derivatives}
  
\begin{enumerate}
  
\item The derivatives of the primal value function in \eqref{eq:vfabs}
at infinity and of the dual value function in \eqref{eq:dvfabs} at
zero are given by
\begin{equation}
u^{\prime}(\infty) := \lim_{x\to\infty}u^{\prime}(x) = 0, \quad
-v^{\prime}(0) := \lim_{y\downarrow 0}(-v^{\prime}(y)) = +\infty.   
\label{eq:mvprime0}
\end{equation}

\item For any fixed $x>0$, with $y=u^{\prime}(x)$ (equivalently
$x=-v^{\prime}(y)$), the primal and dual optimisers
$\widehat{g}(x),\widehat{h}(y)$ are related by
\begin{equation}
U^{\prime}(\widehat{g}(x)) = \widehat{h}(y) =
\widehat{h}(u^{\prime}(x)), \quad \mu\mbox{-a.e.},
\label{eq:primaldual}
\end{equation}
and satisfy
\begin{equation}
\int_{\mathbf{\Omega}}\widehat{g}(x)\widehat{h}(y)\ud\mu = xy =
xu^{\prime}(x).
\label{eq:saturation}
\end{equation}

\item The derivatives of the value functions satisfy the relations
\begin{equation}
xu^{\prime}(x) =
\int_{\mathbf{\Omega}}U^{\prime}(\widehat{g}(x))\widehat{g}(x)\ud\mu,
\quad yv^{\prime}(y) =
\int_{\mathbf{\Omega}}V^{\prime}(\widehat{h}(y))\widehat{h}(y)\ud\mu,
\quad x,y >0.
\label{eq:derivatives}
\end{equation}

\end{enumerate}

\end{lemma}

\begin{proof}

Recall the inequality \eqref{eq:VUbound}, which also applies to the
value functions because they are also conjugate by Lemma
\ref{lem:conjugacy}. We thus have, in addition to \eqref{eq:VUbound},
\begin{equation}
v(y) \geq u(x) -xy, \quad \forall\,x,y>0, \quad \mbox{with equality
iff $y=u^{\prime}(x)$}.
\label{eq:vubound}
\end{equation}
With $\widehat{g}(x)\in\mathcal{C}(x),\,x>0$ and
$\widehat{h}(y)\in\mathcal{D}(y),\,y>0$ denoting the primal and dual
optimisers, the bipolarity relations \eqref{eq:bp1} and \eqref{eq:bp2}
imply that we have
\begin{equation*}
\int_{\mathbf{\Omega}}\widehat{g}(x)\widehat{h}(y)\ud\mu \leq xy, \quad
x,y >0.  
\end{equation*}
Using this as well as \eqref{eq:VUbound} and \eqref{eq:vubound} we
have
\begin{equation}
 0 \leq \int_{\mathbf{\Omega}}\left(V(\widehat{h}(y)) -
U(\widehat{g}(x)) + \widehat{g}(x)\widehat{h}(y)\right)\ud\mu \leq
v(y) - u(x) + xy, \quad x,y>0,
\label{eq:fundineq}
\end{equation}
The right-hand-side of \eqref{eq:fundineq} is zero if and only if
$y=u^{\prime}(x)$, due to \eqref{eq:vubound}, and the non-negative
integrand must then be $\mu$-a.e. zero, which by \eqref{eq:VUbound}
can only happen if \eqref{eq:primaldual} holds, which establishes that
primal-dual relation.

Thus, for any fixed $x>0$ and with $y=u^{\prime}(x)$, and hence
equality in \eqref{eq:fundineq}, we have
\begin{eqnarray*}
0 & = & \int_{\mathbf{\Omega}}\left(V(\widehat{h}(y)) -
U(\widehat{g}(x)) + \widehat{g}(x)\widehat{h}(y)\right)\ud\mu \\
& = & v(y) - u(x) +
\int_{\mathbf{\Omega}}\widehat{g}(x)\widehat{h}(y)\ud\mu \\ 
& = & v(y) - u(x) + xy, \quad y=u^{\prime}(x),
\end{eqnarray*}
which implies that \eqref{eq:saturation} must hold. Inserting the
explicit form of $\widehat{h}(y)=U^{\prime}(\widehat{g}(x))$ into
\eqref{eq:saturation} yields the first relation in
\eqref{eq:derivatives}. Similarly, setting
$\widehat{g}(x)=I(\widehat{h}(y))=-V^{\prime}(\widehat{h}(y))$ into
\eqref{eq:saturation}, with $x=-v^{\prime}(y)$ (equivalent to
$y=u^{\prime}(x)$), yields the second relation in
\eqref{eq:derivatives}.

It remains to establish the relations in \eqref{eq:mvprime0}, which are
equivalent assertions. We shall prove the second one. This will use the
fact that $\mathcal{D}$ is a subset of $L^{1}(\mu)$.

From the second relation in \eqref{eq:derivatives} and the fact that
\begin{equation}
\int_{\mathbf{\Omega}}gh\ud\mu\leq xy, \quad
\forall\,g\in\mathcal{C}(x),h\in\mathcal{D}(y), \quad x,y>0,
\label{eq:absbc}
\end{equation}
we see that, for any $y>0$, we have
$-V^{\prime}(\widehat{h}(y))\in\mathcal{C}(-v^{\prime}(y))$. Thus, for
any $h\in\mathcal{D}$, \eqref{eq:absbc} implies that
\begin{equation}
-v^{\prime}(y) \geq
\int_{\mathbf{\Omega}}-V^{\prime}(\widehat{h}(y))h\ud\mu, \quad
\forall\,h\in\mathcal{D},
\label{eq:uprimeineq}
\end{equation}
which we shall make use of shortly.

Since $\mathcal{D}(y)$ is a subset of $L^{1}(\mu)$, we have
$\int_{\mathbf{\Omega}}\widehat{h}(y)\ud\mu\leq y$, and hence
\begin{equation}
\int_{\mathbf{\Omega}}\frac{\widehat{h}(y)}{y}\ud\mu\leq 1, \quad
\forall\,y>0.
\label{eq:whgxineq}
\end{equation}
Using Fatou's lemma in \eqref{eq:whgxineq} we have
\begin{equation*}
1 \geq \liminf_{y\downarrow 0}
\int_{\mathbf{\Omega}}\frac{\widehat{h}(y)}{y}\ud\mu \geq
\int_{\mathbf{\Omega}}\liminf_{y\downarrow 0}
\left(\frac{\widehat{h}(y)}{y}\right)\ud\mu, 
\end{equation*}
which, given that $\widehat{h}(y)/y$ is non-negative, gives that
$\liminf_{y\downarrow 0}(\widehat{h}(y)/y)<\infty,\,\mu$-a.e.
Therefore, writing $\widehat{h}(y)=:y\widehat{h}^{y}$, which defines a
unique element $\widehat{h}^{y}\in\mathcal{D}$, we have
\begin{equation*}
\widehat{h}^{0} := \liminf_{y\downarrow 0}\widehat{h}^{y} = 
\liminf_{y\downarrow 0}\frac{\widehat{h}(y)}{y} < \infty, \quad
\mu\mbox{-a.e.}   
\end{equation*}
Using this property and applying Fatou's lemma to \eqref{eq:uprimeineq}
we obtain, on using $-V^{\prime}(0)=+\infty$,
\begin{equation*}
+\infty \geq \liminf_{y\downarrow 0}(-v^{\prime}(y)) \geq 
\liminf_{y\downarrow 0}
\int_{\mathbf{\Omega}}-V^{\prime}(y\widehat{h}^{y})h\ud\mu \geq
\int_{\mathbf{\Omega}}\liminf_{y\downarrow 0}
(-V^{\prime}(y\widehat{h}^{y}))h\ud\mu = +\infty, 
\end{equation*}
which gives us the second relation in \eqref{eq:mvprime0}.

\end{proof}

We have now established all results that give the duality in Theorem
\ref{thm:adt}, so let us confirm this.

\begin{proof}[Proof of Theorem \ref{thm:adt}]

Lemma \ref{lem:conjugacy} implies the relations \eqref{eq:conjugacy}
of item (i). The statements in item (ii) are implied by Lemma
\ref{lem:primexis} and Lemma \ref{lem:dualexis}. Items (iii) and (iv)
follow from Lemma \ref{lem:uprime0} and Lemma \ref{lem:derivatives}.

\end{proof}

We are almost ready to prove the concrete duality in Theorem
\ref{thm:itwd}, because Theorem \ref{thm:adt} readily implies nearly
all of the assertions of Theorem \ref{thm:itwd}. The outstanding
assertion is the characterisation of the optimal wealth process in
\eqref{eq:owp} and the associated uniformly integrable martingale
property of the process
$\widehat{M}:=\widehat{X}(x)\widehat{R}(y) +
\int_{0}^{\cdot}\widehat{X}_{s}(x)\widehat{Y}_{s}(y)\ud\kappa_{s}$. So
we proceed to establish these assertions in the proposition below,
which turns out to be interesting in its own right. We take as given
the other assertions of Theorem \ref{thm:itwd}, and in particular the
optimal budget constraint in \eqref{eq:oihbc}. We shall confirm the
proof of Theorem \ref{thm:itwd} in its entirety after the proof of the
next result.

\begin{proposition}[Optimal wealth process]
\label{prop:owp}

Given the saturated budget constraint equality in \eqref{eq:oihbc},
the optimal wealth process is characterised by \eqref{eq:owp}. The
process
\begin{equation*}
\widehat{M}_{t} := \widehat{X}_{t}(x)\widehat{R}_{t}(y) +
\int_{0}^{t}\widehat{X}_{s}(x)\widehat{Y}_{s}(y)\ud\kappa_{s}, \quad
0\leq t <\infty, 
\end{equation*}
is a uniformly integrable martingale, converging to an integrable
random variable $\widehat{M}_{\infty}$, so the martingale extends to
$[0,\infty]$. The process $\widehat{X}(x)\widehat{R}(y)$ is a
potential, that is, a non-negative supermartingale satisfying
$\lim_{t\to\infty}\mathbb{E}[\widehat{X}_{t}(x)\widehat{R}_{t}(y)]=0$.
Moreover, $\widehat{X}_{\infty}(x)\widehat{R}_{\infty}(y)=0$, almost
surely.

\end{proposition}

\begin{proof}

It simplifies notation if we take $x=y=1$, and is without loss of
generality: although $y=u^{\prime}(x)$ in \eqref{eq:oihbc}, one can
always multiply the utility function by an arbitrary constant so as
to ensure that $u^{\prime}(1)=1$.  We thus have the optimal budget
constraint
\begin{equation}
\mathbb{E}\left[\int_{0}^{\infty}
\widehat{X}_{t}\widehat{Y}_{t}\ud\kappa_{t}\right] =1, 
\label{eq:oihbc1}
\end{equation}
for $\widehat{X}\equiv\widehat{X}(1)\in\mathcal{X}$ and
$\widehat{Y}\equiv\widehat{Y}(1)\in\mathcal{Y}$. Since
$\widehat{X}\in\mathcal{X}$, we know there exists an optimal wealth
process $\widehat{X}\equiv\widehat{X}(1)$ and an associated optimal
trading strategy $\widehat{H}$, such that
$ \widehat{X} = 1 + (\widehat{H}\cdot P)\geq 0$, and such that
$\widehat{M}:=\widehat{X}\widehat{R} +
\int_{0}^{\cdot}\widehat{X}_{s}\widehat{Y}_{s}\ud\kappa_{s}$ is a
supermartingale over $[0,\infty)$. The supermartingale condition, by
the same arguments that led to the derivation of the budget constraint
in Lemma \ref{lem:sbc}, leads to the inequality
$\mathbb{E}\left[\int_{0}^{\infty}
  \widehat{X}_{t}\widehat{Y}_{t}\ud\kappa_{t}\right]\leq1$ instead of
the equality \eqref{eq:oihbc1}. Similarly, if the supermartingale is
strict, we get a strict inequality in place of \eqref{eq:oihbc1}. We
thus deduce that $\widehat{M}$ must be a martingale over
$[0,\infty)$. We shall show that this extends to $[0,\infty]$, along
with the other claims in the lemma.

Since $\widehat{M}$ is a martingale, the (non-negative c\`adl\`ag)
deflated wealth process $\widehat{X}\widehat{R}$ is a martingale minus
a non-decreasing process, so is a non-negative c\`adl\`ag
supermartingale, and thus (by Cohen and Elliott \cite[Corollary
5.2.2]{ce15}, for example) converges to an integrable limiting random
variable
$\widehat{X}_{\infty}\widehat{R}_{\infty}:=
\lim_{t\to\infty}\widehat{X}_{t}\widehat{R}_{t}$ (and moreover
$\widehat{X}_{t}\widehat{R}_{t}\geq
\mathbb{E}[\widehat{X}_{\infty}\widehat{R}_{\infty}],\,t\geq 0$). The
non-decreasing integral in $\widehat{M}$ clearly also converges to an
integrable random variable, by virtue of the budget constraint. Thus,
$\widehat{M}$ also converges to an integrable random variable
$\widehat{M}_{\infty}:=\widehat{X}_{\infty}\widehat{R}_{\infty}+
\int_{0}^{\infty}\widehat{X}_{t}\widehat{Y}_{t}\ud\kappa_{t}$. By
Protter \cite[Theorem I.13]{protter}, the extended martingale over
$[0,\infty]$, $(\widehat{M}_{t})_{t\in[0,\infty]}$ is then uniformly
integrable, as claimed.

The martingale condition gives
\begin{equation*}
\mathbb{E}\left[\widehat{X}_{t}\widehat{R}_{t} +
\int_{0}^{t}\widehat{X}_{s}\widehat{Y}_{s}\ud\kappa_{s}\right] = 1,
\quad 0\leq t <\infty. 
\end{equation*}
Taking the limit as $t\to\infty$, using monotone convergence in the
second term within the expectation and utilising \eqref{eq:oihbc1}
yields
\begin{equation*}
\lim_{t\to\infty}\mathbb{E}[\widehat{X}_{t}\widehat{R}_{t}] = 0,
\end{equation*}
so that $\widehat{X}\widehat{R}$ is a potential, as claimed.

Using the uniform integrability of $\widehat{M}$ and taking the limit as
$t\to\infty$ in $\mathbb{E}[\widehat{M}_{t}]=1,\,t\geq 0$, we have
\begin{equation*}
1 = \lim_{t\to\infty}\mathbb{E}[\widehat{M}_{t}] =
\mathbb{E}\left[\lim_{t\to\infty}\widehat{M}_{t}\right] =
\mathbb{E}[\widehat{X}_{\infty}\widehat{R}_{\infty}] + 1,
\end{equation*}
on using \eqref{eq:oihbc1}. Hence, we get
$\mathbb{E}[\widehat{X}_{\infty}\widehat{R}_{\infty}]=0$ and, since
$\widehat{X}_{\infty}\widehat{R}_{\infty}$ is non-negative, we deduce
that $\widehat{X}_{\infty}\widehat{R}_{\infty}=0$, almost surely as
claimed. 

We can now assemble these ingredients to arrive at the optimal wealth
process formula \eqref{eq:owp}. Applying the martingale condition
again, this time over $[t,u]$ for some $t\geq 0$, we have
\begin{equation*}
\mathbb{E}\left[\left.\widehat{X}_{u}\widehat{R}_{u} +
\int_{0}^{u}\widehat{X}_{s}\widehat{Y}_{s}\ud\kappa_{s}
\right\vert\mathcal{F}_{t}\right] = \widehat{X}_{t}\widehat{R}_{t} + 
\int_{0}^{t}\widehat{X}_{s}\widehat{Y}_{s}\ud\kappa_{s}, \quad 0\leq
t\leq u <\infty. 
\end{equation*}
Taking thew limit as $u\to\infty$ and using the uniform integrability
of $\widehat{M}$ we obtain
\begin{equation*}
\mathbb{E}\left[\left.\lim_{u\to\infty}\left(\widehat{X}_{u}\widehat{R}_{u} +
\int_{0}^{u}\widehat{X}_{s}\widehat{Y}_{s}\ud\kappa_{s}\right)
\right\vert\mathcal{F}_{t}\right] = \widehat{X}_{t}\widehat{R}_{t} +
\int_{0}^{t}\widehat{X}_{s}\widehat{Y}_{s}\ud\kappa_{s}, \quad t\geq 0,
\end{equation*}
which, on using $\widehat{X}_{\infty}\widehat{R}_{\infty}=0$,
re-arranges to
\begin{equation*}
\widehat{X}_{t}\widehat{R}_{t} =
\mathbb{E}\left[\left.\int_{t}^{\infty}\widehat{X}_{s}\widehat{Y}_{s}
\ud\kappa_{s}\right\vert\mathcal{F}_{t}\right], \quad t\geq 0, 
\end{equation*}
which establishes \eqref{eq:owp}, and the proof is complete.
  
\end{proof}

\begin{proof}[Proof of Theorem \ref{thm:itwd}]

Given the definitions of the sets $\mathcal{C}(x)$ and
$\mathcal{D}(y)$ in \eqref{eq:Cx} and \eqref{eq:Dy}, respectively, and
the identification of the abstract value functions in \eqref{eq:vfabs}
and \eqref{eq:dvfabs} with their concrete counterparts in
\eqref{eq:primal} and \eqref{eq:dual}, Theorem \ref{thm:adt} implies
all the assertions of Theorem \ref{thm:itwd}, with the exception of
the optimal wealth process formula \eqref{eq:owp} and the uniform
integrability of $\widehat{M}:=\widehat{X}(x)\widehat{R}(y) +
\int_{0}^{\cdot}\widehat{X}_{s}(x)\widehat{Y}_{s}(y)\ud\kappa_{s}$,
which are established by Proposition \ref{prop:owp}.
  
\end{proof}

\section{Examples}
\label{sec:examples}

We end with two examples. The first uses an incomplete market model
with strict local martingale deflators, which is covered in our
framework. The market features a three-dimensional Bessel process for
the market price of risk (MPR) of a stock which also has a stochastic
volatility. We consider the problem \ref{eq:primal} with the measure
$\kappa$ satisfying $\ud\kappa_{t}=\exp(-\alpha t)\ud t$ for a
constant discount rate $\alpha>0$, so that
\begin{equation}
\kappa_{t} = \frac{1}{\alpha}\left(1-\e^{-\alpha t}\right), \quad
t\geq 0.
\label{eq:kalpha}
\end{equation}
Since $\kappa$ is absolutely continuous with respect to
Lebesgue measure, we use the formalism in Remark \ref{rem:sc}. We then
specialise the example to the Black-Scholes model, to confirm that we
obtain results consistent with the example presented by Bouchard and
Pham \cite[Section 4]{bp04}. The market is of course complete in this
simple case.

We shall use a constant relative risk aversion (CRRA) utility function
of the power form:
\begin{equation}
U(x) = \frac{x^{p}}{p}, \quad p<1,\,p\neq 0, \quad x\in\mathbb{R}_{+}.  
\label{eq:crraU}
\end{equation}
The case $p=0$ corresponds formally to logarithmic utility,
$U(x)=\log(x)$, and setting $p=0$ in the results for the power utility
function does indeed recover the results for logarithmic utility, as
can be verified by carrying out the analysis directly for that case.

\begin{example}[Three-dimensional Bessel process MPR, with stochastic
  volatility and correlation]
\label{examp:3db}

Take an infinite horizon complete stochastic basis
$(\Omega,\mathcal{F},\mathbb{F}:=(\mathcal{F}_{t})_{t\geq
  0},\mathbb{P})$, with $\mathbb{F}$ satisfying the usual hypotheses.
Let $(W,W^{\perp})$ be a two-dimensional Brownian motion. We take
$\mathbb{F}$ to be the augmented filtration generated by
$(W,W^{\perp})$.

Let $B$ denote the process which solves the stochastic differential
equation
\begin{equation*}
\ud B_{t} = \frac{1}{B_{t}}\ud t + \ud W_{t} =: \lambda_{t}\ud t + \ud
W_{t}, \quad B_{0}=1.  
\end{equation*}
The process $B$ is the so-called three-dimensional Bessel
process. The process $\lambda:=1/B$ will be the market price of risk
of a stock with price process $P$ and stochastic volatility process
$\sigma>0$, driven by the correlated Brownian motion
$\widetilde{W}:=\rho W+\sqrt{1-\rho^{2}}W^{\perp}$, and with
$\rho\in[-1,1]$ some $\mathbb{F}$-adapted stochastic correlation. We
need not specify the dynamics of $\sigma$ or $\rho$ any further for the
purposes of the example. The stock price dynamics are given by
\begin{equation*}
\ud P_{t} = \sigma_{t}P_{t}\ud B_{t} = \sigma_{t}P_{t}(\lambda_{t}\ud
t + \ud W_{t}). 
\end{equation*}
Note that this model satisfies the so-called \textit{structure condition} of
Pham et al \cite{prs98}, because $P$ admits the decomposition
$P=P_{0}+L+A$ with $L\in\mathcal{M}^{2}_{0,\mathrm{loc}}$ a locally
square-integrable local martingale null at zero and $A$ a predictable
process of finite variation null at zero, and such that
$A=\int_{0}^{\cdot}\widehat{\lambda}_{s}\ud\langle L\rangle_{s}$ for a
predictable process $\widehat{\lambda}$.

Take a constant relative risk aversion (CRRA) utility function as in
\eqref{eq:crraU}, with the measure $\kappa$ given by
\eqref{eq:kalpha}, so that $\gamma_{t}=\e^{\alpha t},\,t\geq 0$. The
primal value function is
\begin{equation*}
u(x) :=
\sup_{X\in\mathcal{X}(x)}\mathbb{E}\left[\int_{0}^{\infty}\e^{-\alpha
t}U(X_{t})\ud t\right], \quad x>0.
\end{equation*}
The wealth process satisfies
\begin{equation}
\ud X_{t} = \sigma_{t}\pi_{t}(\lambda_{t}\ud t + \ud W_{t}),
\quad X_{0}=x,
\label{eq:Xsde}
\end{equation}
where $\pi=HS$ is the trading strategy expressed in terms of the
wealth placed in the stock, with $H$ the process for the number of
shares.

With $\mathcal{E}(\cdot)$ denoting the stochastic exponential, the
supermartingale deflators in this model are given by local martingale
deflators of the form
\begin{equation}
Z:=\mathcal{E}(-\lambda\cdot W - \psi\cdot W^{\perp}),  
\label{eq:Zpsi}
\end{equation}
for an arbitrary process $\psi$ satisfying
$\int_{0}^{t}\psi^{2}_{s}\ud s<\infty$ almost surely for all
$t\geq 0$, with each such $\psi$ leading to a different deflator: this
market is of course incomplete. Let $\Psi$ denote the set of such
integrands $\psi$. In the case that $\sigma$ and $\rho$ are
deterministic, the market is complete and there is a unique local
martingale deflator $Z^{(0)}:=\mathcal{E}(-\lambda\cdot W)$. It is
well-known (see for instance Larsen \cite[Example 2.2]{larsen09}) that
$Z^{(0)}$ is a strict local martingale and, what is more, that
$Z^{(0)}=\lambda$ and that $\lambda$ is square integrable. The strict
local martingale property is inherited by $Z$ in \eqref{eq:Zpsi}, for
any choice of integrand $\psi$.

The supermartingales $R\in\mathcal{R}$ are given by
$R=\exp\left(-\int_{0}^{\cdot}\beta_{s}\ud s\right)Z$ and the
inter-temporal wealth deflators $Y\in\mathcal{Y}$ by $Y=\beta R$,
that is,
\begin{equation}
Y_{t} = \beta_{t}\exp\left(-\int_{0}^{t}\beta_{s}\ud s\right)Z_{t},
\quad t\geq 0,  
\label{eq:wpdexamp}
\end{equation}
with $\beta\in\mathcal{B}$, so $\int_{0}^{\cdot}\beta_{s}\ud s<\infty$
almost surely. The process $M:=XR+\int_{0}^{\cdot}X_{s}Y_{s}\ud s$ is
given as
\begin{equation}
M_{t} := X_{t}R_{t} + \int_{0}^{t}X_{s}Y_{s}\ud s = x +
\int_{0}^{t}R_{s}(\sigma_{s}\pi_{s}-\lambda_{s}X_{s})\ud W_{s} -
\int_{0}^{t}X_{s}R_{s}\psi_{s}\ud W^{\perp}_{s}, \quad
t\geq 0,
\label{eq:M}
\end{equation}
which is a non-negative local martingale and thus a supermartingale.

The convex conjugate of the utility function is
$V(y):=-y^{q}/q,\,y>0$, where $q<1,\,q\neq 0$ is the conjugate
variable to $p$, satisfying $1-q=(1-p)^{-1}$. The dual value function
is given by
\begin{equation*}
v(y) :=
\inf_{Y\in\mathcal{Y}}\mathbb{E}\left[\int_{0}^{\infty}\e^{-\alpha
t}V(yY_{t}\e^{\alpha t})\ud t\right], \quad y>0.
\end{equation*}
The dual minimisation involves both an optimisation over the local
martingale deflators $Z\in\mathcal{Z}$ as well as over the auxiliary
dual control $\beta\in\mathcal{B}$, since the wealth-path deflators
$Y\in\mathcal{Y}$ are given by \eqref{eq:wpdexamp}.

Denote the unique dual minimiser by $\widehat{Y}\in\mathcal{Y}$,
given by
\begin{equation*}
\widehat{Y} = \widehat{\beta}\exp\left(-\int_{0}^{\cdot}
\widehat{\beta}_{s}\ud s\right)\widehat{Z} = \widehat{\beta}\widehat{R},
\end{equation*}
where $\widehat{\beta}\in\mathcal{B}$ is the optimal auxiliary dual
control, $\widehat{R}\in\mathcal{R}$ denotes the optimal incarnation
of the supermartingale $R$ and $\widehat{Z}$ denotes the optimal local
martingale deflator, given by
\begin{equation*}
\widehat{Z} :=\mathcal{E}(-\lambda\cdot W - \widehat{\psi}\cdot
W^{\perp}),
\end{equation*}
for some optimal integrand $\widehat{\psi}$ in \eqref{eq:Zpsi}. For
use below, define the non-negative martingale $H$ by
\begin{equation}
H_{t} := \mathbb{E}\left[\left.
\int_{0}^{\infty}\e^{-\alpha(1-q)s}\widehat{Y}^{q}_{s}\ud
s\right\vert\mathcal{F}_{t}\right], \quad t\geq 0. 
\label{eq:H}
\end{equation}

Using \eqref{eq:pditwprime}, the optimal wealth process is given by
\begin{equation}
(\widehat{X}_{t}(x))^{-(1-p)} = u^{\prime}(x)\e^{\alpha
t}\widehat{Y}_{t}, \quad t\geq 0.
\label{eq:Xhat1}
\end{equation}
By \eqref{eq:oihbcprime} the optimisers satisfy the saturated budget
constraint
\begin{equation}
\mathbb{E}\left[\int_{0}^{\infty}\widehat{X}_{t}(x)\widehat{Y}_{t}\ud
t\right]=x.
\label{eq:sbc1}
\end{equation}
The relations \eqref{eq:Xhat1} and \eqref{eq:sbc1} yield
\begin{equation}
\widehat{X}_{t}(x) =
\frac{x}{H_{0}}\e^{-\alpha(1-q)t}\widehat{Y}^{-(1-q)}_{t}, \quad t\geq
0.
\label{eq:Xhat2}
\end{equation}
Using the result \eqref{eq:Xhat2} in the right-hand-side of
\eqref{eq:owpprime}, the optimal wealth process then also satisfies
\begin{equation*}
\widehat{X}_{t}(x)\widehat{R}_{t} = \frac{x}{H_{0}}\mathbb{E}\left[\left.
\int_{t}^{\infty}\e^{-\alpha(1-q)s}\widehat{Y}^{q}_{s}\ud
s\right\vert\mathcal{F}_{t}\right], \quad t\geq 0.  
\end{equation*}
More pertinently, the optimal martingale $\widehat{M}$, corresponding
to the process in \eqref{eq:M} at the optimum, is computed as
\begin{equation}
\widehat{M}_{t} := \widehat{X}_{t}(x)\widehat{R}_{t} +
\int_{0}^{t}\widehat{X}_{s}(x)\widehat{Y}_{s}\ud s =
\frac{x}{H_{0}}H_{t}, \quad t\geq 0,
\label{eq:Mhat}
\end{equation}
so is indeed a martingale.

By martingale representation, $\widehat{M}$ will have a stochastic
integral representation which, without loss of generality, can be
written in the form
\begin{equation}
\widehat{M}_{t} = x +
\int_{0}^{t}\widehat{R}_{s}\widehat{X}_{s}(x)(\varphi_{s} 
-q\lambda_{s})\ud W_{s} +
\int_{0}^{t}\widehat{R}_{s}\widehat{X}_{s}(x)\xi_{s}\ud
W^{\perp}_{s}, \quad t\geq 0, 
\label{eq:Mhatsi}
\end{equation}
for some integrands $\varphi,\xi$. Comparing with the representation
in \eqref{eq:M} at the optimum yields the optimal trading strategy in
terms of the optimal portfolio proportion
$\widehat{\theta}:=\widehat{\pi}/\widehat{X}(x)$, and the optimal
integrand $\widehat{\psi}$, as
\begin{equation}
\widehat{\theta}_{t} := \frac{\widehat{\pi}_{t}}{\widehat{X}_{t}(x)} =
\frac{\lambda_{t}}{\sigma_{t}(1-p)} + \frac{\varphi_{t}}{\sigma_{t}}, \quad
\widehat{\psi}_{t} = -\xi_{t}, \quad t\geq 0.  
\label{eq:thetahat}
\end{equation}
In particular, the process $\varphi$ records the correction to the
Merton-type strategy $\lambda/(\sigma(1-p))$ due to the stochastic
volatility and correlation.

This is as far as one can go without computing explicitly the dual
minimiser $\widehat{Y}$, which is typically impossible in closed form
for power utility, except for some special cases such as a
Black-Scholes model (as we shall show further below).

For the special case of logarithmic utility, one can set $p=0$ and
$q=0$ in the results for power utility, which gives that $H=1/\alpha$
is constant, and so $\widehat{M}=x$ is also constant, yielding
\begin{equation*}
\widehat{\theta}_{t} = \frac{\lambda_{t}}{\sigma_{t}}, \quad
\widehat{\psi}_{t} = 0, \quad t\geq 0,
\end{equation*}
giving the classic myopic trading strategy for logarithmic utility
(and the correction to the Merton strategy satisfies
$\varphi=q\lambda=0$ for $q=0$, as it should).

In particular, since $\widehat{\psi}\equiv 0$, the dual optimiser is
given as
\begin{equation}
\widehat{Y} = \widehat{\beta}\exp\left(-\int_{0}^{\cdot}
\widehat{\beta}_{s}\ud s\right)Z^{(0)},  
\label{eq:Yhat}
\end{equation}
for some optimal auxiliary dual control
$\widehat{\beta}\in\mathcal{B}$, with
$Z^{(0)}=\mathcal{E}(-\lambda\cdot W)$ the minimal local martingale
deflator. Moreover, setting $q=0$ in \eqref{eq:Xhat2} and using
$H=1/\alpha$ gives the optimal wealth process in the form
\begin{equation}
\widehat{X}_{t}(x) = \frac{\alpha x\e^{-\alpha t}}{\widehat{Y}_{t}},
\quad t\geq 0.
\label{eq:Xhat3}
\end{equation}
But, using the optimal strategy
$\widehat{\pi}=(\lambda/\sigma)\widehat{X}(x)$ in the wealth SDE
\eqref{eq:Xsde}, we also compute that
\begin{equation}
\widehat{X}_{t}(x) = \frac{x}{Z^{(0)}_{t}}, \quad t\geq 0.  
\label{eq:Xhat4}
\end{equation}
Equating the two expressions for $\widehat{X}(x)$ in \eqref{eq:Xhat3}
and \eqref{eq:Xhat4}, and then using \eqref{eq:Yhat}, yields that the
optimal auxiliary dual control is also constant, and given by
\begin{equation}
\widehat{\beta}_{t} = \alpha, \quad t\geq 0.  
\label{eq:betahat}
\end{equation}

These results for logarithmic utility can of course be obtained by
going directly through the analysis from scratch in the manner
above. Indeed, one can directly compute the dual value function, as
follows. Using the defintion \eqref{eq:dualprime} along with
$V(y)=-(1+\log(y))$ for logarithmic utility, one expresses the dual
value function as
\begin{equation*}
v(y) = \frac{1}{\alpha}\left(V(y)-1\right) +
\inf_{\beta\in\mathcal{B},\psi\in\Psi}\mathbb{E}\left[\int_{0}^{\infty}
\e^{-\alpha t}\left(\int_{0}^{t}(\beta_{s} +
\frac{1}{2}(\lambda^{2}_{s} + \psi^{2}_{s}))\ud s -
\log(\beta_{t})\right)\ud t\right].
\end{equation*}
The optimisations over $\psi$ and $\beta$ can be carried out
separately. Clearly, the term involving $\psi$ is minimised by
$\widehat{\psi}\equiv 0$, while an integration by parts in the
remaining integrals yields
\begin{equation*}
v(y) = \frac{1}{\alpha}\left(V(y)-1\right) +
\inf_{\beta\in\mathcal{B}}\mathbb{E}\left[\int_{0}^{\infty}
\e^{-\alpha t}\left(\frac{1}{2\alpha}\lambda^{2}_{t} +
\frac{\beta_{t}}{\alpha} - \log(\beta_{t})\right)\ud t\right].
\end{equation*}
The minimisation over $\beta$ can then be carried out pointwise,
yielding \eqref{eq:betahat} and giving the dual optimiser for
logarithmic utility:
$\widehat{Y}_{t}=\alpha\exp\left(-\alpha t\right)Z^{(0)}_{t},\,t\geq
0$ as before. Using this dual optimiser in \eqref{eq:pditwprime} gives
\eqref{eq:Xhat4}.

\end{example}

\begin{example}[Black-Scholes model, CRRA utility]
\label{examp:bs}

If we specialise Example \ref{examp:3db} to the case where $\lambda$
and $\sigma$ are constant, we are in a Black-Scholes market and the
computations for power utility can be carried out explicitly. We show
this in order to verify that our formalism reproduces the results of
the example in Bouchard and Pham \cite[Section 4]{bp04}. The market is
now complete, and there is a unique local martingale deflator given by
$Z=\mathcal{E}(-\lambda W)$. The wealth-path deflators take the form
\begin{equation*}
Y = \beta\exp\left(-\int_{0}^{\cdot}\beta_{s}\ud s\right)Z,  
\end{equation*}
for some $\beta\in\mathcal{B}$.

With this structure, the same method as for Example \ref{examp:3db}
yields the same representation \eqref{eq:Xhat2} for the optimal wealth
process, where in this case the dual minimiser is given by
\begin{equation}
\widehat{Y} =
\widehat{\beta}\exp\left(-\int_{0}^{\cdot}\widehat{\beta}_{s}\ud
s\right) \mathcal{E}(-\lambda W),   
\label{eq:Yhatbs}
\end{equation}
for some optimal auxiliary dual control
$\widehat{\beta}\in\mathcal{B}$, and the martingale $H$ in
\eqref{eq:H} has the same representation with the dual minimiser
in \eqref{eq:Yhatbs} in place. 

The process $M$ of \eqref{eq:M} is this time given by the same
expression but without the integral involving $\psi$, so we have
\begin{equation*}
M_{t} := X_{t}R_{t} + \int_{0}^{t}X_{s}Y_{s}\ud s = x +
\int_{0}^{t}R_{s}(\sigma\pi_{s}-\lambda X_{s})\ud W_{s}, \quad
t\geq 0.
\end{equation*}
The optimal martingale $\widehat{M}$ once again has the representation
in \eqref{eq:Mhat}, and has a stochastic integral representation of
the form in \eqref{eq:Mhatsi} but without the integral with respect to
$W^{\perp}$, and we once again find an expression of the form in
\eqref{eq:thetahat} for the optimal trading strategy. Our goal is to
now compute the dual minimiser, by computing $\widehat{\beta}$, and to
thus show that the correction $\varphi$ to the Merton strategy is zero
in this case.

To compute $\widehat{\beta}$ we examine the dual value function, which
is expressed in the form
\begin{equation*}
v(y) =
\inf_{\beta\in\mathcal{B}}V(y)\mathbb{E}\left[\int_{0}^{\infty}
\exp\left(-\alpha(1-q)t-q\int_{0}^{t}\beta_{s}\ud
s\right)\beta^{q}_{t}Z^{q}_{t}\ud t\right]. 
\end{equation*}
Given the constant parameters of the model, one now makes the (not
unreasonable) ansatz that $\widehat{\beta}$ is deterministic, and in
fact constant. With this conjecture, one passes the expectation inside
the integral, uses
\begin{equation}
\mathbb{E}\left[\left.Z^{q}_{u}\right\vert\mathcal{F}_{t}\right] =
\mathcal{E}(-q\lambda
W)_{t}\exp\left(-\frac{1}{2}q(1-q)\lambda^{2}u\right), \quad 0\leq
t\leq u,
\label{eq:ezq}
\end{equation}
and computes the resultant expression to arrive at
\begin{equation*}
v(y) =
\inf_{\beta}V(y)\left(\frac{\beta^{q}}{q\beta + (1-q)(\alpha +
\frac{1}{2}q\lambda^{2})}\right). 
\end{equation*}  
Straightforward differentiation gives the (constant) optimiser as
\begin{equation*}
\widehat{\beta} = \alpha +\frac{1}{2}q\lambda^{2},
\end{equation*}
and \eqref{eq:Yhatbs} then gives the dual minimiser. With this in
place, one expresses the martingale $H$ in the form
\begin{equation*}
H_{t} = \left(\alpha
+\frac{1}{2}q\lambda^{2}\right)^{q}\mathbb{E}\left[\left.
\int_{0}^{\infty}\exp\left(\left(\alpha
+\frac{1}{2}q\lambda^{2}\right)
u\right)Z^{q}_{u}\right\vert\mathcal{F}_{t}\right],
\quad t\geq 0.  
\end{equation*}
Once again, we take the expectation inside the integral and use
\eqref{eq:ezq}, and we arrive at
\begin{equation*}
H_{t} = \left(\alpha +
\frac{1}{2}q\lambda^{2}\right)^{-(1-q)}\mathcal{E}(-q\lambda W)_{t},
\quad t\geq 0.
\end{equation*}
This in turn yields that the optimal martingale $\widehat{M}$ is given
by
\begin{equation*}
\widehat{M}_{t} =  x\frac{H_{t}}{H_{0}} = x \mathcal{E}(-q\lambda W)_{t},
\quad t\geq 0,
\end{equation*}
and the optimal wealth process is given by the representation
\eqref{eq:Xhat2} as
\begin{equation*}
\widehat{X}_{t}(x) = x\frac{\mathcal{E}(-q\lambda W)_{t}}{Z_{t}},
\quad t\geq 0.  
\end{equation*}
Thus, the processes $\widehat{M},\widehat{X}(x)$ are related according
to
\begin{equation*}
\widehat{X}_{t}(x)Z_{t} = \widehat{M}_{t}, \quad t\geq 0.  
\end{equation*}
We can now compute the optimal trading strategy. Using the dynamics of
the wealth process for any strategy $\pi$, given by \eqref{eq:Xsde}
with constant parameters, we have that
\begin{equation}
X_{t}Z_{t} = x + \int_{0}^{t}\left(\sigma\pi_{s}-\lambda
X_{s}\right)\ud W_{s}, \quad t\geq 0.  
\label{eq:XZ1}
\end{equation}
On the other hand, at the optimum, since
$\widehat{X}(x)Z=x\mathcal{E}(-q\lambda W)$, we have
\begin{equation}
\widehat{X}_{t}(x)Z_{t} = x -
q\lambda\int_{0}^{t}\widehat{X}_{s}Z_{s}\ud W_{s}, \quad t\geq 0. 
\label{eq:XZ2}
\end{equation}
Equating \eqref{eq:XZ1} at the optimum with \eqref{eq:XZ2} gives the
optimal trading strategy as
\begin{equation*}
\widehat{\theta}_{t} \equiv
\frac{\widehat{\pi}_{t}}{\widehat{X}_{t}(x)} = \frac{\lambda}{\sigma(1-p)}, \quad t\geq 0,
\end{equation*}
so the optimal strategy is the Merton strategy, as expected.

\end{example}

{\small
\bibliography{ihumfitw_refs}
\bibliographystyle{siam}
}

\end{document}